%% file: main.tex
\documentclass[letterpaper,journal]{IEEEtran}
\usepackage{amsmath,amssymb,amsfonts}
\usepackage{mathtools}
\usepackage{algorithmic}
\usepackage{algorithm}
\usepackage{array}
\usepackage{placeins}
\usepackage[caption=false,font=normalsize,labelfont=sf,textfont=sf]{subfig}
\usepackage{textcomp}
\usepackage{stfloats}
\usepackage{url}
\usepackage{verbatim}
\usepackage{graphicx}
\usepackage{cite}
\usepackage{float}
\newtheorem{theorem}{Theorem}
\newtheorem{lemma}[theorem]{Lemma}
\newtheorem{proposition}[theorem]{Proposition}

\begin{document}

\title{Robust Multi-Source DoA Estimation under Weather-Induced Distortion: Structured Covariance Methods and Performance Bounds}

\author{Chenyang~Yan,~\IEEEmembership{Graduate Student Member,~IEEE,}
        Geert~Leus,~\IEEEmembership{Fellow,~IEEE,}
        and~Mats~Bengtsson,~\IEEEmembership{Senior Member,~IEEE}

\thanks{This work has been funded by the FP6 -- FutuRe project, from the Europe's Rail Joint Undertaking under the European Union’s Horizon 2020 research and innovation programme, grant agreement No. 101101962. This publication reflects only the author’s view and the EuRail JU is not responsible for any use that may be made of the
information it contains.}
}


\IEEEpubid{}

\maketitle

\begin{abstract}
Direction-of-arrival (DoA) estimation in adverse weather is degraded by propagation-induced phase and amplitude distortions that violate the covariance structure assumed by classical subspace methods. Motivated by a physics-based model of rain propagation, we develop a structured covariance formulation for rain-distorted arrays in both single- and multi-source scenarios. A key ingredient is to relax the physically parameterized distortion model and instead model the distortion covariance as an unknown real-valued Hermitian Toeplitz matrix. For the single-source case, we derive a closed-form covariance-matching calibration algorithm and provide structured-covariance and physics-informed Cram\'{e}r--Rao lower bounds (CRLBs). For the multi-source case, we prove that the original per-source distortion model is non-identifiable and introduce a collinear approximation across angles. Under this approximation, we propose three multi-source calibration methods: alternating LASSO, joint LASSO, and a joint nuclear-norm formulation. Simulations under heavy rain demonstrate improved DoA accuracy and source resolvability compared with conventional baselines, highlighting the benefit of structured covariance modeling and calibration for safety-critical sensing.
\end{abstract}

\input{introduction}

\input{distortion_model}

\input{measurement_model}
\input{doa_covariance_matching}
\input{calibration}

\input{performance_bounds}

\input{numerical_results}

\input{conclusion}
\bibliographystyle{IEEEtran}
\bibliography{ref}
\input{appendix}
\end{document}

%% file: introduction.tex
\section{Introduction}
Direction-of-arrival (DoA) estimation is a fundamental problem in array signal processing, with direct relevance to automotive sensing, intelligent transportation, and safety monitoring at critical infrastructures such as railway level crossings \cite{yan2026obstacle}. Classical spectral and subspace methods, including Capon beamforming, MUSIC, Root-MUSIC, and ESPRIT, provide well-established benchmarks under narrowband and well-calibrated array models \cite{Capon1969,Schmidt1986,Barabell1983,RoyKailath1989}. Modern mmWave radar systems have made these techniques practical for real-time localization and tracking \cite{Hasch2012}, while robust beamforming, covariance fitting, sparse reconstruction, and gridless super-resolution further improve angular discrimination when the assumed covariance structure is reliable \cite{Vorobyov2013,Malioutov2005,Candes2014,VanTrees2002,stoica2005spectral}.

Adverse weather violates this favorable structure. Besides attenuation and backscattering, rain can induce stochastic phase and amplitude distortions because raindrop sizes are comparable to mmWave radar wavelengths. Previous studies have identified rain as a particularly severe weather condition for radar propagation \cite{Ballistic}, and related random-media models also arise for snow and mist \cite{pozhidaev2010estimation}. We therefore focus on rainy conditions, while the proposed covariance-based formulation is applicable to broader random-media scenarios.

Most weather-aware radar studies quantify degradation through attenuation coefficients or SINR penalties, whereas the impact of random media on the array covariance is less often modeled explicitly. Recent physics-based propagation models based on the $S$-matrix method describe rain as a sequence of random scattering slabs and preserve phase relationships that are critical for DoA estimation \cite{yektakhah2024model,yektakhah2023physics}. These models show that weather-induced distortions can be viewed as multiplicative array-response perturbations, which explains the failure of classical subspace methods and motivates structured covariance calibration.

Building on this insight, this paper develops a robust single- and multi-source DoA estimation framework that treats weather distortion as a structured nuisance. Instead of relying directly on a fully parameterized physical distortion model, we represent the distortion covariance by an unknown real-valued Toeplitz matrix. This yields a covariance-domain formulation that separates the ideal signal covariance from the rain-induced component and enables generalized least-squares calibration. In the single-source case, the resulting estimator admits an efficient closed-form solution. In the multi-source case, we show that the original angle-dependent distortion model can become non-identifiable, introduce a reduced collinear approximation, and propose three calibration algorithms based on alternating LASSO, joint LASSO, and nuclear-norm regularization. We further derive Cramér--Rao lower bounds (CRLBs) using the Slepian--Bangs formula for complex Gaussian data \cite{stoica2005spectral}. These bounds quantify how the nuisance parameters associated with weather-induced distortion reduce the effective Fisher information for DoA estimation, and how this reduction depends on rain severity, array aperture, and integration time. Numerical results under realistic rain rates demonstrate improved DoA accuracy and source resolvability compared with conventional baselines. 

\textit{Contributions.} The main contributions are summarized below. This journal manuscript substantially extends our prior conference paper~\cite{yan2025robustcovariance} in several directions:
\begin{itemize}
    \item We generalize the rain-induced distortion model and the associated covariance formulation from the single-source setting to the multi-source case.

    \item We develop a covariance-matching estimation framework that admits closed-form least-squares (LS) / weighted least-squares (WLS) calibration in the single-source case. For the multi-source case, we propose three calibration algorithms that leverage sparse and low-rank covariance structure: an alternating-optimization method with LASSO, a joint convex LASSO formulation, and a joint nuclear-norm regularized formulation. 

    \item We derive and compare the angle-dependent Cramér--Rao lower bounds (CRLBs) for both the proposed structured model and the physics-based model in the single-source scenario, and we assess the consistency between the empirical RMSE and the theoretical bounds. In addition, we identify conditions under which multi-source estimation becomes non-identifiable in the absence of additional structural constraints.
\end{itemize}
\FloatBarrier

%% file: distortion_model.tex
\section{Distortion Model under Adverse Weather}
\label{sec:2}

\begin{table}[t]
\centering
\caption{Parameter settings and $\alpha$ values for Figs.~\ref{pdf_demo} and~\ref{ratio}.}
\begin{tabular}{c|c|c|c|c}
\hline
Case & $d$ $(\lambda_0)$ & $R$ (m) & Rain rate (mm/hr) & $\alpha$ \\
\hline
(i)  & 4 & 200 & 25 & 0.6470 \\[2pt]
(ii) & 4 & 400 & 25 & 0.6217 \\[2pt]
(iii)& 8 & 200 & 25 & 0.5598 \\[2pt]
(iv) & 4 & 200 & 50 & 0.4994 \\[2pt]
\hline
\end{tabular}
\label{parameter_table}
\end{table}

Following the rain-induced wavefront fluctuation model in~\cite{yektakhah2024model}, let $P$ denote a point on the received wavefront. For a given rain realization, let $E(P)\in\mathbb{C}$ denote the complex electric-field envelope at $P$. The mean (coherent) field is defined as $E_{\mathrm{mean}}(P)\triangleq \mathbb{E}\{E(P)\}$, where $\mathbb{E}\{\cdot\}$ denotes the ensemble average over independent rain realizations. The normalized electric-field fluctuation is defined as
\begin{equation}
E_n(P) \triangleq \frac{E(P)-E_{\mathrm{mean}}(P)}{E_{\mathrm{mean}}(P)}.
\label{eq:normalized_field}
\end{equation}

Consider two points $P_1$ and $P_2$ located on the same wavefront, at a distance $R$ (in meter) from the source, and separated by a distance $d$ (in wavelength). The corresponding normalized fluctuations, denoted by $E_{n,1}$ and $E_{n,2}$, are modeled as zero-mean, jointly circularly symmetric complex Gaussian random variables with second-order statistics
\begin{align}
\mathbb{E}\{|E_{n,1}|^2\} &= \mathbb{E}\{|E_{n,2}|^2\} = 2\lambda_{11}, 
\\
\mathbb{E}\{E_{n,1} E_{n,2}^*\} &= 2\alpha\lambda_{11}, \quad 
\Im\!\left\{\mathbb{E}\{E_{n,1} E_{n,2}^*\}\right\}=0,
\label{eq:corss_correlation}
\end{align}
where $\lambda_{11}$ controls the fluctuation power and the real valued coefficient $\alpha$ characterizes the spatial correlation on the wavefront. This parameter $\alpha$ can be computed using the empirical model in~\cite[(13)]{yektakhah2024model}:
\begin{equation}
\alpha 
= \exp\!\left(
-a_1 \left(\frac{R}{a_2 R + 1}\right)
\left(\frac{d}{a_3d + 1}\right)
\right),
\label{eq:empirical_alpha}
\end{equation}
where the empirical coefficients $a_1$, $a_2$, and $a_3$ depend on the rain rate and operating frequency, and their values taken from~\cite[Table~II]{yektakhah2024model}. This empirical model is reported to be valid for $R \leq 500\,\mathrm{m}$ and $0.1 \leq d \leq 8$~\cite{yektakhah2024model}.

\begin{figure}[t]
\centering
\includegraphics[width=\columnwidth]{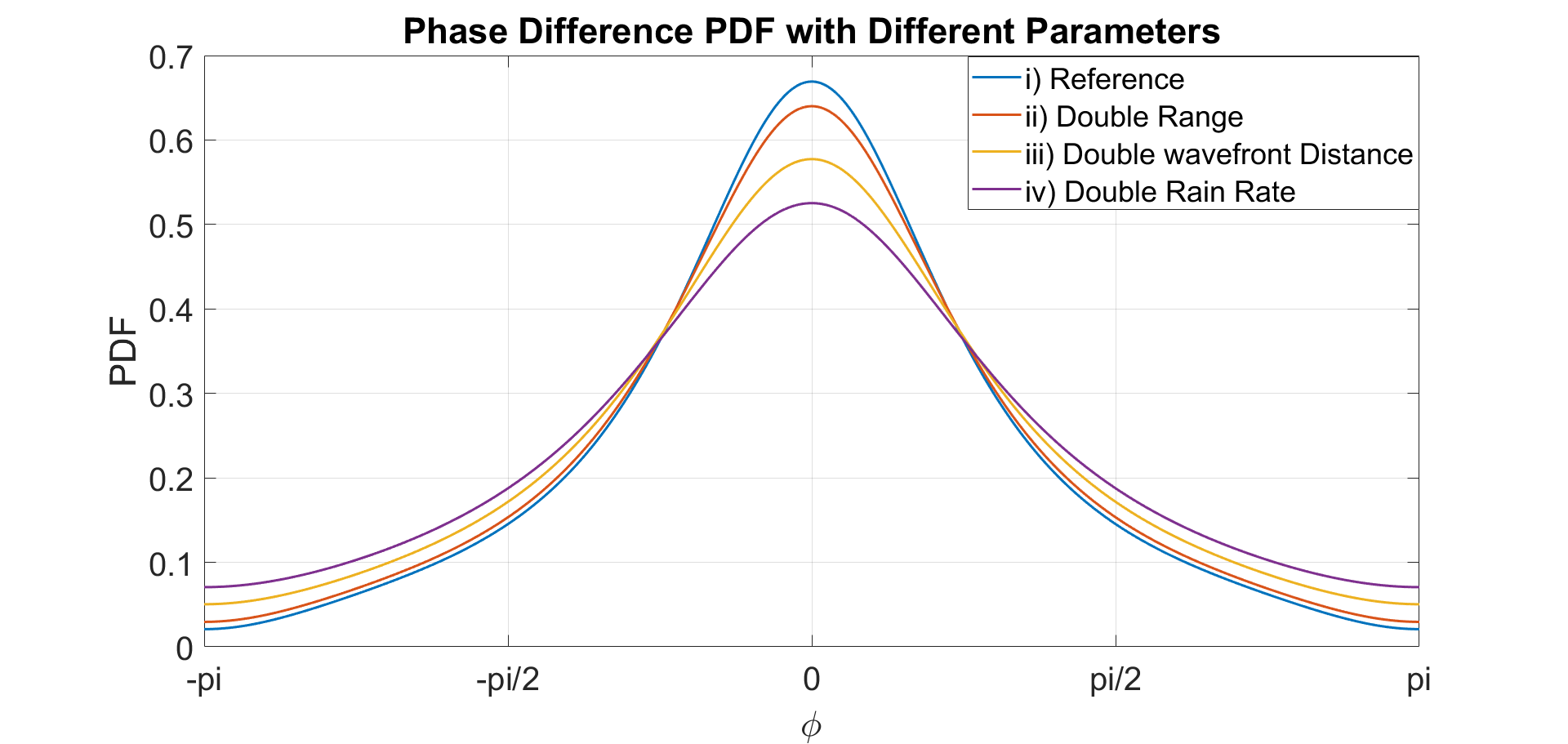}
\caption{Phase-difference pdf under different propagation conditions.}
\label{pdf_demo}
\end{figure}

\begin{figure}[t]
\centering
\includegraphics[width=\columnwidth]{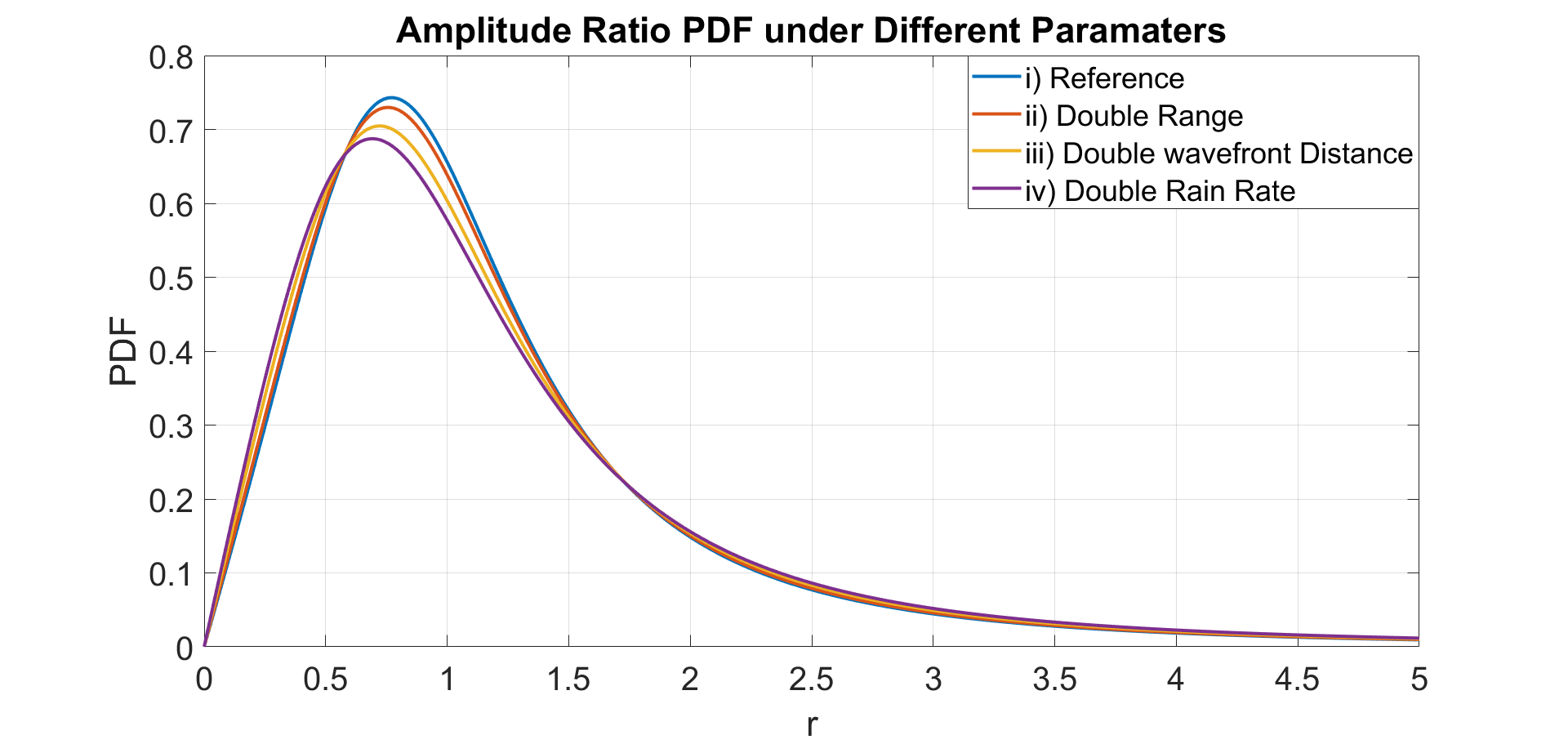}
\caption{Magnitude-ratio pdf under different propagation conditions.}
\label{ratio}
\end{figure}

Under this setup, and following~\cite{yektakhah2024model,sarabandi1992derivation}, simplified probability density functions (pdfs) for the phase difference $\phi=\angle\{E_1\}-\angle\{E_2\}$ and the magnitude ratio $r=|E_1/E_2|$ between two points separated by $d$ are given in~\cite[(9a), (9b)]{yektakhah2024model}. Figures~\ref{pdf_demo} and~\ref{ratio} illustrate these pdfs under different propagation conditions. Four cases are evaluated, with parameters summarized in Table~\ref{parameter_table}. The reference case (i) uses $d=4$, $R=200\,\mathrm{m}$, and a rain rate of $25\,\mathrm{mm/hr}$. In case (ii), the range is increased to $400\,\mathrm{m}$; in case (iii), the separation is increased to $8$; and in case (iv), the rain rate is increased to $50\,\mathrm{mm/hr}$, while the remaining parameters are held fixed.

As $R$, $d$, or the rain rate increases, both phase and amplitude fluctuations become more pronounced. Specifically, the phase-difference pdf becomes less concentrated around $0^\circ$, and the magnitude-ratio pdf shifts away from unity. Moreover, the phase-difference pdfs remain symmetric about zero, consistent with the assumed real-valued cross-correlation in~\eqref{eq:corss_correlation}.

%% file: measurement_model.tex
\section{Measurement Model}
\label{sec:meas_model}

We assume far-field propagation and consider a uniform linear array (ULA) with $M$ antennas and inter-element spacing $d_0$ measured in wavelengths. 
As illustrated in Fig.~\ref{array_model}, the rain-induced electric-field fluctuations introduced in Sec.~\ref{sec:2} are assumed to occur across the incident plane wave when it reaches the first array element encountered along the propagation path. 
We refer to this incident field as the \emph{reference plane wavefront}. 
Fluctuations beyond this reference plane are assumed negligible. 
Moreover, to ensure the validity of the empirical correlation model in~\eqref{eq:empirical_alpha}, the total array aperture is assumed not to exceed $8\lambda_0$~\cite{yektakhah2024model}, where $\lambda_0$ is the wavelength of the narrowband signal. 
This corresponds to at most 17 antennas under half-wavelength spacing.

\subsection{Snapshot-Domain Model}

We first connect the wavefront model in Sec.~\ref{sec:2} to array observations. 
Let $P_m$ denote the location of the $m$th sensor, and identify the received baseband signal $y_m(t)$ with the complex electric field $E(P_m,t)$. 
Substituting~\eqref{eq:normalized_field} yields
\begin{equation}
y_m(t)
=
\bigl(1+E_n(P_m,t)\bigr)E_{\mathrm{mean}}(P_m),
\label{eq:ym_rain}
\end{equation}
where $E_{\mathrm{mean}}(P_m)$ denotes the coherent, distortion-free field contribution at $P_m$. 
Accordingly, we define the multiplicative distortion at sensor $m$ as
\begin{equation}
b_m(t)
\triangleq
1+E_n(P_m,t),
\label{eq:b_m}
\end{equation}
which is consistent with the assumption in~\cite{yektakhah2024model}. This convention differs from our previous implementation~\cite{yan2025robustcovariance}, where the distortion term was taken as $b_m(t)=E_n(P_m,t)$.

\begin{figure}[t]
\centering
\includegraphics[width=\columnwidth]{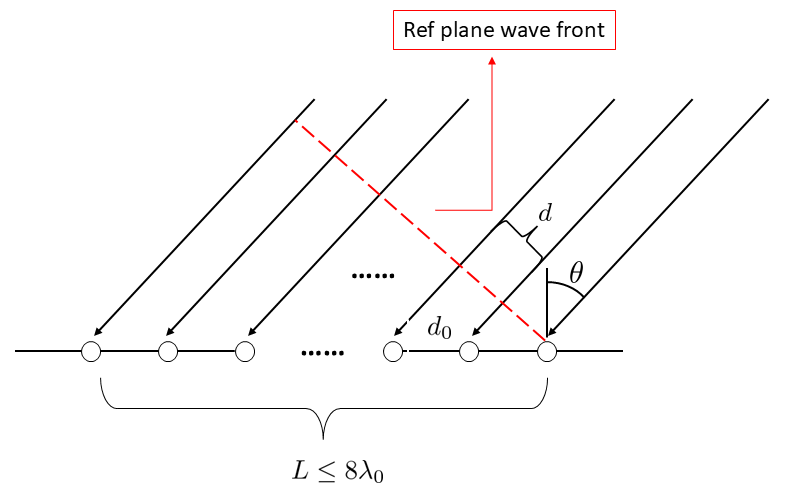}
\caption{Illustration of the rain-distorted ULA measurement model and the reference plane wavefront.}
\label{array_model}
\end{figure}

For a narrowband source signal $s(t)$ impinging from direction $\theta$, the received signal at the $m$th array element is modeled as
\begin{equation}
y_m(t)
=
\beta\,s(t)\,a_m(\theta)b_m(t)+n_m(t),
\label{eq:rx_m1}
\end{equation}
where $\beta$ is the complex path-loss coefficient, including both the range-dependent propagation loss and the additional attenuation induced by rain along the propagation path, $a_m(\theta)=e^{j2\pi (m-1)d_0\sin(\theta)}$, $m=1,\dots,M$, is the $m$th entry of the distortion-free steering vector, and $n_m(t)$ is Gaussian and spatially white additive noise. 
Writing $b_m(t)=|b_m(t)|e^{j\phi_m(t)}$,~\eqref{eq:rx_m1} can equivalently be expressed as
\begin{equation}
y_m(t)
=
\beta\,s(t)|b_m(t)|
e^{j\left(2\pi (m-1)d_0\sin(\theta)+\phi_m(t)\right)}
+n_m(t),
\label{eq:rx_m2}
\end{equation}
where $|b_m(t)|$ and $\phi_m(t)$ denote the distortion amplitude and phase, respectively. 
Both are assumed to vary across snapshots, reflecting the random-medium nature of the propagation. 
We further assume that, for each sensor $m$, $\{b_m(t)\}$ are independent and identically distributed across snapshots.

Stacking the sensor outputs gives the single-source vector model
\begin{equation}
\mathbf{y}(t)
=
\beta s(t)
\bigl(\mathbf{a}(\theta)\odot\mathbf{b}(t)\bigr)
+\mathbf{n}(t),
\label{eq:concise_receive_model}
\end{equation}
where $\mathbf{y}(t)\in\mathbb{C}^{M}$, $\mathbf{n}(t)\in\mathbb{C}^{M}$, $\odot$ denotes the Hadamard product,
\begin{equation}
\mathbf{a}(\theta)
=
[1,e^{j2\pi d_0\sin(\theta)},\dots,e^{j2\pi(M-1)d_0\sin(\theta)}]^T,
\label{eq:steering_vector}
\end{equation}
and $\mathbf{b}(t)=[b_1(t),\dots,b_M(t)]^T$.
We also define the unit-power steering covariance
\begin{equation}
\mathbf{S}(\theta)
\triangleq
\mathbf{a}(\theta)\mathbf{a}(\theta)^H,
\label{eq:S_def}
\end{equation}
whose entries have unit modulus and whose phase encodes the DoA.

We now extend~\eqref{eq:concise_receive_model} to $N$ mutually uncorrelated sources located at directions $\{\theta_n\}_{n=1}^{N}$. 
The rain-distorted multi-source snapshot model is
\begin{equation}
\mathbf{y}(t)
=
\sum_{n=1}^{N}
\beta_n s_n(t)
\bigl(\mathbf{a}(\theta_n)\odot\mathbf{b}(\theta_n,t)\bigr)
+\mathbf{n}(t),
\label{eq:multi_measurement}
\end{equation}
where $s_n(t)$ denotes the $n$th source signal, $\beta_n$ is its complex path-loss coefficient, and $\mathbf{b}(\theta_n,t)\in\mathbb{C}^{M}$ is the distortion vector associated with the wavefront arriving from $\theta_n$. 

The distortion is angle-dependent because the effective sensor-pair separation on the reference plane wavefront varies with incidence angle. 
For an element separation $|m-\ell|d_0$ along the array axis, the corresponding reference-plane wavefront separation is
\begin{equation}
d_{|m-\ell|}(\theta_n)
=
|m-\ell|d_0\cos(\theta_n).
\label{eq:d_theta}
\end{equation}
Consequently, the correlation parameter $\alpha$ in~\eqref{eq:empirical_alpha}, and thus the statistics of $\mathbf{b}(\theta_n,t)$, depend on $\theta_n$ through $d_{|m-\ell|}(\theta_n)$.

\subsection{Covariance-Domain Model}
For estimation, we will primarily work in the covariance domain. Assume that the sources are zero-mean, mutually uncorrelated, and independent of the noise, with
\[
\mathbf{R}_s=\mathbb{E}\{\mathbf{s}(t)\mathbf{s}(t)^H\}
=\mathrm{diag}(\sigma_1^2,\dots,\sigma_N^2),
\]
where $\mathbf{s}(t)=[s_1(t),\dots,s_N(t)]^T$, and $\sigma_n^2=\mathbb{E}\{|s_n(t)|^2\}$ denotes the power of the $n$th source. Also assume that $\mathbf{n}(t)$ is zero-mean with covariance $\mathbf{R}_n=\mathbb{E}\{\mathbf{n}(t)\mathbf{n}(t)^H\}$.
Then the covariance of~\eqref{eq:multi_measurement} can be written as
\begin{equation}
\mathbf{R}_y=\sum_{n=1}^{N}|\beta_n|^2\bigl(\mathbf{R}_x(\theta_n)\odot\mathbf{R}_b(\theta_n)\bigr)+\mathbf{R}_n,
\label{eq:covariance_measurement}
\end{equation}
where the distortion-free rank-one source covariance for the $n$th source is
\begin{equation}
\mathbf{R}_x(\theta_n)
=
\sigma_n^2\mathbf{S}(\theta_n)
=
\sigma_n^2\,\mathbf{a}(\theta_n)\mathbf{a}(\theta_n)^H,
\label{eq:Rx_def}
\end{equation}
and $\mathbf{S}(\theta_n)$ is the corresponding unit-power steering covariance defined in~\eqref{eq:S_def}. Moreover, $\mathbf{R}_b(\theta_n)\in\mathbb{C}^{M\times M}$ denotes the rain-induced distortion covariance associated with direction $\theta_n$, defined entry-wise by
\begin{equation}
\bigl[\mathbf{R}_b(\theta_n)\bigr]_{m\ell}\triangleq
\mathbb{E}\bigl\{b_{n,m}(t)\,b_{n,\ell}^*(t)\bigr\},
\qquad
1\le m,\ell\le M.
\label{eq:Rb_def}
\end{equation}

Using~\eqref{eq:b_m} together with the second-order statistics in Sec.~\ref{sec:2}, the distortion covariance admits the form
\begin{equation}
\bigl[\mathbf{R}_b(\theta_n)\bigr]_{m\ell}=1+\mathbb{E}\{E_{n}(P_m,t)\,E_{n}(P_\ell,t)^*\}.
\label{eq:Rb_from_En}
\end{equation}
In particular, under the model in Sec.~\ref{sec:2},
\begin{equation}
\bigl[\mathbf{R}_b(\theta_n)\bigr]_{m\ell}=
\begin{cases}
1+2\lambda_{11}, & m=\ell,\\[2pt]
1+2\lambda_{11}\,\alpha_{|m-\ell|}(\theta_n), & m\neq \ell,
\end{cases}
\label{eq:Rb_piecewise}
\end{equation}
where $\alpha_{|m-\ell|}(\theta_n)$ is evaluated by~\eqref{eq:empirical_alpha} using the effective separation on the reference wavefront of \eqref{eq:d_theta}. Therefore, $\mathbf{R}_b(\theta_n)$ depends on $\theta_n$ through the angle-dependent separation, and it carries the rain-rate dependence through the empirical parameters $(a_1,a_2,a_3)$ in~\eqref{eq:empirical_alpha}. This dependence propagates to the observed covariance $\mathbf{R}_y$ in~\eqref{eq:covariance_measurement}, forming the basis for the structured covariance-matching estimators developed in the next section.

%% file: doa_covariance_matching.tex
\section{DoA Estimation Using Covariance Matching}
\label{sec:covariance_matching}

Generalized least squares (GLS), also known as covariance-matching estimation techniques (COMET), provides a principled framework for matching the sample covariance matrix with a parametric covariance model~\cite{kariya2004generalized,OTTERSTEN1998185}. 
This approach is widely used in array signal processing to exploit second-order statistics for parameter estimation.

Consider the sample covariance $\hat{\mathbf{R}}_y$ formed from $T$ independent snapshots of the array output $\mathbf{y}(t)\in\mathbb{C}^{M}$. 
Our goal is to estimate the DoAs $\boldsymbol{\theta}=[\theta_1,\dots,\theta_N]^T$ from $\hat{\mathbf{R}}_y$, while treating the weather-induced distortion covariances as nuisance parameters.

Starting from the covariance model in~\eqref{eq:covariance_measurement}, the GLS criterion can be written as
\begin{equation}
\begin{aligned}
&(\hat{\boldsymbol{\theta}},
\{\hat{\mathbf{R}}_{b,n}\}_{n=1}^{N},
\hat{\boldsymbol{\beta}},
\hat{\boldsymbol{\sigma}}^{2})
=
\operatorname*{arg\,min}_{\boldsymbol{\theta},\,\{\mathbf{R}_{b,n}\}_{n=1}^{N},\,\boldsymbol{\beta},\,\boldsymbol{\sigma}^{2}}
\\
&\quad
\left\|
\hat{\mathbf{R}}_y
-
\sum_{n=1}^{N}
|\beta_n|^2
\bigl[
\mathbf{R}_{x}(\theta_n)\odot\mathbf{R}_{b,n}
\bigr]
\right\|_{\mathbf{W}}^{2},
\end{aligned}
\label{eq:GLS_cost_explicit_beta}
\end{equation}
where $\mathbf{R}_{b,n}\triangleq\mathbf{R}_b(\theta_n)$, $\boldsymbol{\beta}=[\beta_1,\dots,\beta_N]^T$, $\boldsymbol{\sigma}^{2}=[\sigma_1^2,\dots,\sigma_N^2]^T$, and $\mathbf{R}_{x}(\theta_n)$ is defined in~\eqref{eq:Rx_def}. 
The weighted norm is defined as
\begin{equation}
\lVert \mathbf{X} \rVert_{\mathbf{W}}^{2}
=
\mathrm{vec}(\mathbf{X})^{H}\mathbf{W}^{-1}\mathrm{vec}(\mathbf{X}).
\label{eq:Wnorm_def}
\end{equation}
Since the covariance depends on $\beta_n$ and $\sigma_n^2$ only through the product $|\beta_n|^2\sigma_n^2$, these quantities cannot be separated from covariance data alone. 
We therefore use different reparameterizations in the multi- and single-source cases.

For the multi-source case, the source-dependent path loss is absorbed into the corresponding source covariance. 
Define the attenuated source signal and power as
\begin{equation}
\tilde{s}_n(t)
\triangleq
\beta_n s_n(t),
\qquad
\tilde{\sigma}_n^2
\triangleq
\mathbb{E}\{|\tilde{s}_n(t)|^2\}
=
|\beta_n|^2\sigma_n^2,
\label{eq:stilde_sigma_tilde_def}
\end{equation}
and the attenuated source covariance as
\begin{equation}
\tilde{\mathbf{R}}_{x}(\theta_n)
\triangleq
\tilde{\sigma}_n^2
\mathbf{S}(\theta_n)
=
|\beta_n|^2\mathbf{R}_{x}(\theta_n).
\label{eq:Rxtilde_def}
\end{equation}
Then~\eqref{eq:GLS_cost_explicit_beta} becomes
\begin{equation}
\begin{aligned}
&(\hat{\boldsymbol{\theta}},
\{\hat{\mathbf{R}}_{b,n}\}_{n=1}^{N},
\hat{\tilde{\boldsymbol{\sigma}}}^{2})
=
\operatorname*{arg\,min}_{\boldsymbol{\theta},\,\{\mathbf{R}_{b,n}\}_{n=1}^{N},\,\tilde{\boldsymbol{\sigma}}^{2}}
\\
&\quad
\left\|
\hat{\mathbf{R}}_y
-
\sum_{n=1}^{N}
\bigl[
\tilde{\mathbf{R}}_x(\theta_n)\odot\mathbf{R}_{b,n}
\bigr]
\right\|_{\mathbf{W}}^{2},
\end{aligned}
\label{eq:GLS_cost}
\end{equation}
where $\tilde{\boldsymbol{\sigma}}^2=[\tilde{\sigma}_1^2,\dots,\tilde{\sigma}_N^2]^T$. 

For the single-source case, the same scalar ambiguity can instead be absorbed into the distortion nuisance term. 
Using $\mathbf{R}_x(\theta)=\sigma^2\mathbf{S}(\theta)$, the single-source signal covariance contribution satisfies
\begin{equation}
|\beta|^2
\bigl[
\mathbf{R}_{x}(\theta)\odot\mathbf{R}_{b}
\bigr]
=
\mathbf{S}(\theta)
\odot
\tilde{\mathbf{R}}_{b},
\label{eq:single_beta_absorption}
\end{equation}
where
\begin{equation}
\tilde{\mathbf{R}}_{b}
\triangleq
|\beta|^2\sigma^2\mathbf{R}_{b}.
\label{eq:Rbtilde_def}
\end{equation}
Here $\tilde{\mathbf{R}}_{b}$ is a scaled distortion covariance that includes both the path loss and the source power. 
The single-source GLS criterion therefore reduces to
\begin{equation}
(\hat{\theta},\hat{\tilde{\mathbf{R}}}_{b})
=
\arg\min_{\theta,\tilde{\mathbf{R}}_{b}}
\left\|
\hat{\mathbf{R}}_y
-
\mathbf{S}(\theta)
\odot
\tilde{\mathbf{R}}_{b}
\right\|_{\mathbf{W}}^{2}.
\label{eq:WGLS_single_cost}
\end{equation}

In the large-sample regime, the asymptotically optimal weighting is given by~\cite{OTTERSTEN1998185}
\begin{equation}
\hat{\mathbf{W}}
=
\hat{\mathbf{R}}_{y}^{T}\otimes\hat{\mathbf{R}}_{y},
\label{eq:opt_weight}
\end{equation}
where $\otimes$ denotes the Kronecker product.

To rewrite~\eqref{eq:GLS_cost} in vectorized WLS form, define
\begin{equation}
\hat{\mathbf{r}}
=
\mathrm{vec}(\hat{\mathbf{R}}_{y}),
\label{eq:rhat_covmatch}
\end{equation}
and
\begin{equation}
\begin{aligned}
&\mathbf{r}_{\rm MS}
\!\left(
\boldsymbol{\theta},
\tilde{\boldsymbol{\sigma}}^2,
\{\mathbf{R}_{b,n}\}_{n=1}^{N}
\right)
=\\ &
\mathrm{vec}\!\left(
\sum_{n=1}^{N}
\left[
\tilde{\mathbf{R}}_x(\theta_n)\odot\mathbf{R}_{b,n}
\right]
\right).
\label{eq:r_ms_def}  
\end{aligned}
\end{equation}
The multi-source WLS criterion becomes
\begin{align}
(\hat{\boldsymbol{\theta}},
\hat{\tilde{\boldsymbol{\sigma}}}^{2},
\{\hat{\mathbf{R}}_{b,n}\}_{n=1}^{N})
&=
\operatorname*{arg\,min}_{\boldsymbol{\theta},\,\tilde{\boldsymbol{\sigma}}^{2}\ge\mathbf{0},\,\{\mathbf{R}_{b,n}\}_{n=1}^{N}} \\&
\bigl(
\hat{\mathbf{r}}-\mathbf{r}_{\rm MS}
\bigr)^{H}
\hat{\mathbf{W}}^{-1}
\bigl(
\hat{\mathbf{r}}-\mathbf{r}_{\rm MS}
\bigr).
\label{eq:WGLS_quadform}
\end{align}

For the single-source case, define
\begin{equation}
\mathbf{r}_{\rm SS}(\theta,\tilde{\mathbf{R}}_b)
=
\mathrm{vec}
\left(
\mathbf{S}(\theta)
\odot
\tilde{\mathbf{R}}_b
\right).
\label{eq:r_ss_def}
\end{equation}
Then~\eqref{eq:WGLS_single_cost} can be written as
\begin{align}
(\hat{\theta},\hat{\tilde{\mathbf{R}}}_{b})
&=
\arg\min_{\theta,\tilde{\mathbf{R}}_{b}}
\bigl(
\hat{\mathbf{r}}
-
\mathbf{r}_{\rm SS}(\theta,\tilde{\mathbf{R}}_{b})
\bigr)^{H}
\hat{\mathbf{W}}^{-1}
\nonumber\\[-1mm]
&\hspace{22mm}\times
\bigl(
\hat{\mathbf{r}}
-
\mathbf{r}_{\rm SS}(\theta,\tilde{\mathbf{R}}_{b})
\bigr).
\label{eq:WGLS_single_quadform}
\end{align}

As a simpler special case, choosing $\mathbf{W}=\mathbf{I}$ yields the unweighted LS criterion. 

The unweighted formulation implicitly treats all covariance entries as equally reliable, whereas $\hat{\mathbf{R}}_y$ generally exhibits a nonuniform variance structure. 
The weighting in~\eqref{eq:opt_weight} accounts for this heteroscedasticity and is asymptotically equivalent to maximum-likelihood weighting under Gaussian measurements~\cite{OTTERSTEN1998185}. 
While WLS often improves estimation accuracy relative to LS, it increases computational cost due to the construction and inversion of $\hat{\mathbf{W}}$ and the associated matrix--vector operations.

In the next section, we develop algorithmic strategies for solving the WLS and LS optimization problems efficiently for both the single- and multi-source cases.

%% file: calibration.tex
\section{Calibration Algorithms}
\label{sec:calibration}
In this section, we address the optimization problem in~\eqref{eq:GLS_cost} for estimating the DoAs $\hat{\boldsymbol{\theta}}$. Here, we first generalize the single-source procedure from~\cite{yan2025robustcovariance} to a weighted LS (WLS) criterion, and then develop a new approach that extends the calibration and estimation framework to the multi-source scenario.

\subsection{Single-Source Calibration}

For a ULA, the unit-power steering covariance $\mathbf{S}(\theta)=\mathbf{a}(\theta)\mathbf{a}(\theta)^H$ possesses a positive semidefinite Hermitian Toeplitz (HT) structure and has unit-modulus entries. In our approach, we relax the physically parameterized distortion model in \eqref{eq:corss_correlation}--\eqref{eq:b_m} and instead assume only that the scaled distortion covariance $\tilde{\mathbf{R}}_{b}$ is an (unknown) real-valued Toeplitz matrix. Since both $\mathbf{S}(\theta)$ and $\tilde{\mathbf{R}}_{b}$ are Toeplitz, their Hadamard product $\mathbf{R}_{\mathrm{HT}}=\mathbf{S}(\theta)\odot\tilde{\mathbf{R}}_{b}$ is also HT. The proposed calibration method therefore begins by estimating the combined matrix $\mathbf{R}_{\mathrm{HT}}$ and subsequently decouples it by exploiting that $\mathbf{S}(\theta)$ has unit-modulus complex entries carrying phase information, whereas $\tilde{\mathbf{R}}_{b}$ is real-valued Toeplitz. Once $\mathbf{S}(\theta)$ is recovered, subspace-based DoA estimators such as MUSIC can be applied.

Because $\mathbf{R}_{\mathrm{HT}}$ is HT and positive semidefinite, it can be completely described by $(2M-1)$ real-valued parameters. In particular,
\begin{equation}
\begin{aligned}
\mathbf{R}_{\mathrm{HT}}
&= \sum_{m=0}^{2M-2} c_m \boldsymbol{\Sigma}_m \\
&= c_0 \mathbf{I}_M
+ \sum_{m=1}^{M-1}
\left(
c_m \mathbf{T}_m
+
c_{m+M-1}\tilde{\mathbf{T}}_m
\right),
\end{aligned}
\label{R_decomposition}
\end{equation}
where $\mathbf{T}_m$ contains ones on subdiagonals $\pm m$ and zeros elsewhere, and $\tilde{\mathbf{T}}_m$ contains $j$ on subdiagonal $+m$ and $-j$ on subdiagonal $-m$~\cite{romero2015compression,shan1985spatial}.

Substituting~\eqref{R_decomposition} into the WLS criterion and defining
\begin{equation}
\mathbf{c} = [\,c_0,\,c_1,\,\dots,\,c_{2M-2}\,]^{T},
\end{equation}
we obtain
\begin{equation}
\hat{\mathbf{c}}_{\mathrm{WLS}}
=
\arg\min_{\mathbf{c}}
\left\|
\hat{\mathbf{R}}_{y}-\sum_{m=0}^{2M-2} c_m \boldsymbol{\Sigma}_m
\right\|_{\hat{\bf W}}^{2}.
\label{eq:c_wls_mat}
\end{equation}
Let
\begin{equation}
\hat{\mathbf{r}}=\mathrm{vec}(\hat{\mathbf{R}}_{y}),
\label{eq:rhat_def}
\end{equation}
and
\begin{equation}
\mathbf{V}
=
\big[
\mathrm{vec}(\boldsymbol{\Sigma}_0),\,
\mathrm{vec}(\boldsymbol{\Sigma}_1),\,
\dots,\,
\mathrm{vec}(\boldsymbol{\Sigma}_{2M-2})
\big],
\label{eq:V_def}
\end{equation}
where $\mathrm{vec}(\cdot)$ stacks the columns of its matrix argument into a vector. Then~\eqref{eq:c_wls_mat} is equivalent to
\begin{equation}
\hat{\mathbf{c}}_{\mathrm{WLS}}
=
\arg\min_{\mathbf{c}}
\left(
\hat{\mathbf{r}}-\mathbf{V}\mathbf{c}
\right)^{H}
\hat{\mathbf{W}}^{-1}
\left(
\hat{\mathbf{r}}-\mathbf{V}\mathbf{c}
\right),
\label{eq:c_wls_vec}
\end{equation}
with the closed-form solution:
\begin{equation}
\hat{\mathbf{c}}_{\mathrm{WLS}}
=
\left(\mathbf{V}^{H}\hat{\mathbf{W}}^{-1}\mathbf{V}\right)^{-1}
\mathbf{V}^{H}\hat{\mathbf{W}}^{-1}\hat{\mathbf{r}}.
\label{eq:c_wls_sol}
\end{equation}

The unweighted LS estimator is obtained from~\eqref{eq:c_wls_vec} by setting $\hat{\mathbf{W}}=\mathbf{I}$, yielding:
\begin{equation}
\hat{\mathbf{c}}_{\mathrm{LS}}
=
\left(\mathbf{V}^{H}\mathbf{V}\right)^{-1}\mathbf{V}^{H}\hat{\mathbf{r}}.
\label{eq:c_ls_sol}
\end{equation}

Once $\hat{\mathbf{c}}_{\rm{LS}}$ (or $\hat{\mathbf{c}}_{\rm{WLS}}$) is obtained, the estimated HT matrix $\hat{\mathbf{R}}_{\rm{HT}}$ is reconstructed via~\eqref{R_decomposition}. The decoupling step then separates the phase and magnitude components:
\begin{equation}
\hat{\mathbf{S}} = \exp\!\big(j\,\angle(\hat{\mathbf{R}}_{\rm{HT}})\big),
\end{equation}
\begin{equation}
\hat{\tilde{\mathbf{R}}}_b = |\hat{\mathbf{R}}_{\rm{HT}}|.
\end{equation}

The overall GLS-based calibration procedure is summarized in Algorithm~\ref{algorithm1}; depending on the choice of the weighting matrix, the update step can be implemented using either LS or WLS.

\begin{algorithm}[htb]
\caption{Single-Source Calibration}
\begin{algorithmic}[1]
\STATE \textbf{Input}: Sample covariance $\hat{\mathbf{R}}_y$, basis matrices $\{\boldsymbol{\Sigma}_m\}_{m=0}^{2M-2}$.
\STATE Form the basis matrix
\[
\mathbf{V}
=
\big[
\mathrm{vec}(\boldsymbol{\Sigma}_0),\dots,
\mathrm{vec}(\boldsymbol{\Sigma}_{2M-2})
\big].
\]
\STATE Choose the weighting matrix $\hat{\mathbf{W}}$:
\[
\hat{\mathbf{W}} =
\begin{cases}
\mathbf{I}, & \text{(LS)}\\[2pt]
\hat{\mathbf{R}}_y^{T} \otimes \hat{\mathbf{R}}_y, & \text{(WLS)}
\end{cases}.
\]
\STATE Compute the coefficient vector
\[
\hat{\mathbf{c}}
=
\big(\mathbf{V}^H \hat{\mathbf{W}}^{-1}\mathbf{V}\big)^{-1}
\mathbf{V}^H \hat{\mathbf{W}}^{-1} \,\hat{\mathbf{r}}.
\]
\STATE Reconstruct the Toeplitz matrix
\[
\hat{\mathbf{R}}_{\rm{HT}} = \sum_{m=0}^{2M-2} \hat{c}_m \boldsymbol{\Sigma}_m.
\]
\STATE Obtain the scaled distortion covariance and steering covariance estimates
\[
\hat{\tilde{\mathbf{R}}}_b = |\hat{\mathbf{R}}_{\rm{HT}}|, 
\qquad
\hat{\mathbf{S}} = \exp\!\big(j\,\angle(\hat{\mathbf{R}}_{\rm{HT}})\big).
\]
\end{algorithmic}
\label{algorithm1}
\end{algorithm}

\subsection{Multiple-Source Calibration}
\label{subsec:m-sources}

\begin{proposition}
\label{prop:multi_identifiability}
Consider the covariance model in~\eqref{eq:covariance_measurement}. For $N=1$, if $\mathbf{R}_b(\theta)$ is real-valued Toeplitz, and its scale is fixed, e.g., by a diagonal normalization, then the DoA and the distortion covariance are identifiable up to the source-power scaling. 
For $N\geq 2$, if each source is assigned an independent angle-dependent distortion covariance $\mathbf{R}_b(\theta_n)$, the model is in general not identifiable from $\mathbf{R}_y$ alone.
\end{proposition}

\begin{IEEEproof}
For a single source, the real-valued Toeplitz structure of $\mathbf{R}_b(\theta)$ separates the phase progression induced by the steering covariance from the real distortion lags; after fixing the scale ambiguity, the DoA and the distortion coefficients can be recovered. 
In the multi-source case, however, $\mathbf{R}_y$ remains Hermitian Toeplitz and therefore contains only $2M-1$ real degrees of freedom, whereas the $N$ independent real-valued Toeplitz matrices $\{\mathbf{R}_b(\theta_n)\}_{n=1}^{N}$ already introduce $NM$ nuisance parameters, in addition to the source powers and DoAs. 
Hence the parameter-to-covariance mapping is generally non-injective unless additional structure, such as a common physical parameterization or the approximation in~\eqref{eq:colinar_specific}, is imposed.
\end{IEEEproof}

\subsubsection{Collinear Approximation for Multi-Source Calibration}
In practice, retaining the angle-dependent form of $\mathbf{R}_b$ makes the problem extremely challenging. Moreover, since the angle information within $\mathbf{R}_b$ has not been explicitly exploited, it is natural to ask how much variation $\mathbf{R}_b$ exhibits across different angles. To quantify the similarity between different $\mathbf{R}_b$ matrices, we introduce the cosine similarity coefficient:
\begin{equation}
\rho_{ij}
= \frac{\mathrm{vec}[\mathbf{R}_b(\theta_i)]^H  \mathrm{vec}[\mathbf{R}_b(\theta_j)]}{\lVert \mathbf{R}_b(\theta_i) \rVert_F \, \lVert \mathbf{R}_b(\theta_j) \rVert_F}.
\end{equation}
By scanning over different incident angles, we can evaluate $\rho_{ij}$ across angle pairs to assess the degree of similarity among the distortion matrices.

Figure~\ref{fig:cosim_heatmap} shows the heatmap of the cosine-similarity coefficient for all angle pairs in $[-90^\circ,\,90^\circ]$, under a rain rate of $50\,\mathrm{mm/hr}$, a propagation range of $200\,\mathrm{m}$, and an array with $M=8$ antennas. The cosine similarity exceeds $0.998$ for all angle pairs, indicating that the vectorized distortion covariance matrices are nearly collinear across the entire angular range. The result suggests that the two
vectors are nearly aligned in the Frobenius vector space. 
This motivates the proportional (collinear) approximation:
\begin{equation}
\mathbf{R}_b(\theta_i) \approx p_{ij}\,\mathbf{R}_b(\theta_j),
\label{eq:colinear_general}
\end{equation}
where $p_{ij}\in\mathbb{R}$ denotes an angle-pair-dependent proportionality coefficient.

By taking $\mathbf{R}_b(0^\circ)$, the distortion matrix corresponding to $0^\circ$, as a reference, distortion matrices of the incident angles can be approximated as scalar multiples of this reference matrix:
\begin{equation}
    \mathbf{R}_b(\theta_i) \approx p_{i}\,\mathbf{R}_b(0^ \circ), \quad i=1,\dots,N.
\label{eq:colinar_specific}
\end{equation}

Substituting~\eqref{eq:colinar_specific} into the covariance expression~\eqref{eq:covariance_measurement} yields:
\begin{equation}
\begin{aligned}
    \mathbf{R}_y
    &\approx \mathbf{R}_b(0^ \circ) \odot \Biggl(\sum_{n=1}^{N}p_n|\beta_n|^2\mathbf{R}_x(\theta_n)\Biggr) + \mathbf{R}_n.
\label{eq:cov_virtual_b}
\end{aligned}
\end{equation}

Using $\mathbf{R}_x(\theta_n)=\sigma_n^2\mathbf{S}(\theta_n)$, define the mixed undistorted covariance as
\begin{equation}
\begin{aligned}
    \mathbf{R}_{x,\mathrm{mix}}
    &\triangleq
    \sum_{n=1}^{N}p_n|\beta_n|^2\mathbf{R}_x(\theta_n) \\
    &=
    \sum_{n=1}^{N}p_n |\beta_n|^2\sigma_n^2\,\mathbf{S}(\theta_n).
\end{aligned}
\label{eq:Rx_mix_continuous}
\end{equation}
The scalar weights $p_n$ are not estimated separately. Instead, together with the path-loss and source-power factors, they are absorbed into the angular grid power vector introduced below.

Let $\{\vartheta_g\}_{g=1}^{G}$ be an angular grid spanning the field of view and define
\begin{equation}
\mathbf{A}=\big[\mathbf{a}(\vartheta_1),\,\dots,\,\mathbf{a}(\vartheta_G)\big].
\label{eq:grid_array_response}
\end{equation}
The mixed covariance is represented on this grid as
\begin{equation}
\mathbf{R}_x(\boldsymbol{\gamma})
=
\mathbf{A}\,\mathrm{diag}(\boldsymbol{\gamma})\,\mathbf{A}^{H},
\label{eq:Rx_gamma_def}
\end{equation}
where $\boldsymbol{\gamma}\in\mathbb{R}_{+}^{G}$ directly absorbs the unknown products $p_n|\beta_n|^2\sigma_n^2$. Thus, if a source at $\theta_n$ lies on the grid point $\vartheta_g$, then $\gamma_g=p_n|\beta_n|^2\sigma_n^2$, while inactive grid points have zero power.

\begin{figure}[t]
    \centering
    \includegraphics[width=\linewidth]{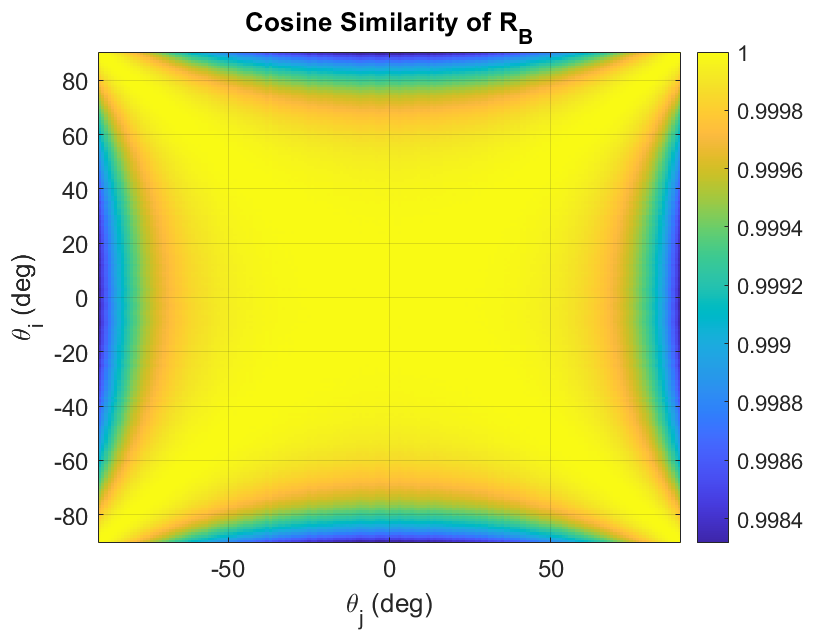}
    \caption{Heatmap of cosine similarity over all angle pairs.}
    \label{fig:cosim_heatmap}
\end{figure}

Therefore, the approximate multi-source covariance-matching problem can be written directly in terms of the grid power vector as
\begin{equation}
\begin{aligned}
(\hat{\boldsymbol{\gamma}},\hat{\mathbf{R}}_b(0^\circ))
&=
\operatorname*{arg\,min}_{\boldsymbol{\gamma}\ge\mathbf{0},\,\mathbf{R}_b(0^\circ)} \\
&
\left\|
\hat{\mathbf{R}}_{y}
-
\mathbf{R}_b(0^\circ)\odot \mathbf{R}_{x}(\boldsymbol{\gamma})
\right\|_{\hat{\mathbf{W}}}^{2}.
\label{eq:GLS_approx}
\end{aligned}
\end{equation}
The active entries of $\boldsymbol{\gamma}$ determine the estimated DoAs, while their magnitudes represent the effective reweighted powers.

The entries of $\mathbf{R}_{x,\mathrm{mix}}$, or equivalently $\mathbf{R}_{x}(\boldsymbol{\gamma})$, are no longer unit-modulus. Consequently, the phase--magnitude separation exploited in the single-source calibration is no longer applicable.

A viable alternative is to adopt a sparse covariance representation that absorbs this reweighting into a sparse power vector defined on an angular grid, thereby enabling sparse covariance reconstruction \cite{tibshirani1996regression,yang2018sparse}. In addition, the low-rank nature of $\mathbf{R}_{x,\mathrm{mix}}$ motivates complementary formulations based on nuclear-norm regularization and related structural constraints. Accordingly, we develop three algorithms for multi-source calibration: (i) an alternating-optimization method with an $\ell_1$-regularized (LASSO) power update, (ii) a joint LASSO-based formulation, and (iii) a joint formulation with nuclear-norm regularization.

These three methods are motivated by different trade-offs between fidelity to the original covariance-matching criterion and computational complexity. The alternating LASSO method remains closest to the original GLS/WLS objective by updating the power spectrum and distortion parameters in separate, structurally consistent subproblems; however, this iterative block-coordinate strategy can be computationally demanding. By contrast, the two joint formulations modify the original cost by introducing a distortion-compensated representation, which enables the power vector and distortion parameters to be estimated simultaneously within a single convex program. This joint structure typically reduces the number of outer iterations and simplifies implementation, at the expense of relying on an altered mismatch criterion that is more sensitive to modeling mismatch.

\subsubsection{Alternating-Optimization Method with LASSO}
Using the grid representation in~\eqref{eq:grid_array_response}--\eqref{eq:Rx_gamma_def}, the vector $\boldsymbol{\gamma}=[\gamma_1,\dots,\gamma_G]^T$ is expected to be sparse, and the indices of its nonzero entries indicate the active DoAs. Its nonzero magnitudes directly represent the effective reweighted powers $p_n|\beta_n|^2\sigma_n^2$ on the active grid points.

As before, we parameterize the distortion covariance using its real-valued Toeplitz structure. Using the reference direction $0 ^\circ$, we write
\begin{equation}
\mathbf{R}_b(0 ^\circ;\mathbf{c})
=
\sum_{m=0}^{M-1} c_m\,\mathbf{T}_m,
\end{equation}
where $\mathbf{c}=[c_0,\dots,c_{M-1}]^T$ and $\{\mathbf{T}_m\}$ are real-valued Toeplitz basis matrices.

To remove the scale ambiguity between $\mathbf{R}_x(\boldsymbol{\gamma})$ and $\mathbf{R}_b(0 ^\circ;\mathbf{c})$, we fix the zero-lag (diagonal) entry of the distortion covariance, i.e.,
\begin{equation}
\big[\mathbf{R}_b(0 ^\circ;\mathbf{c})\big]_{m,m} = c_0 = 1.
\label{eq:rb_diag_fix}
\end{equation}
In addition, we impose the natural non-negativity constraints $\boldsymbol{\gamma}\ge\mathbf{0}$ and $\mathbf{c}\ge\mathbf{0}$.

A LASSO-regularized formulation of~\eqref{eq:GLS_approx} is then
\begin{equation}
\begin{aligned}
(\hat{\boldsymbol{\gamma}},\hat{\mathbf c})&= \operatorname*{arg\,min}_{\boldsymbol{\gamma}\ge 0,\;\mathbf c \ge 0} \\&
\left\|
\hat{\mathbf{R}}_y
-
\mathbf{R}_x(\boldsymbol{\gamma})\odot \mathbf{R}_b(0 ^\circ;\mathbf{c})
\right\|_F^2
+\mu\left\| \boldsymbol{\gamma}\right\|_1,
\\
&\text{s.t.}\quad
c_0=1,
\label{eq:multi_lasso_alt}
\end{aligned}
\end{equation}
where $\mu>0$ controls the sparsity level of $\boldsymbol{\gamma}$.

With the above parameterizations, the cost in~\eqref{eq:multi_lasso_alt} is bi-convex with respect to $\boldsymbol{\gamma}$ and $\mathbf{c}$, which motivates an alternating-optimization strategy. Specifically, at iteration $k$, we update $\boldsymbol{\gamma}$ and $\mathbf{c}$ in turn while keeping the other variable fixed.

Step 1a ($\boldsymbol{\gamma}$-update):
For fixed $\mathbf{c}^{(k)}$, define
\begin{equation}
\mathbf{R}_b^{(k)} \triangleq \mathbf{R}_b(0 ^\circ;\mathbf{c}^{(k)}),
\end{equation}
and compute $\boldsymbol{\gamma}^{(k+1)}$ via
\begin{equation}
\begin{aligned}
\boldsymbol{\gamma}^{(k+1)}
&=
\arg\min_{\boldsymbol{\gamma}\ge \mathbf{0}} \\
&\left\|
\hat{\mathbf{R}}_y
-
\big(\mathbf{A}\,\mathrm{diag}(\boldsymbol{\gamma})\,\mathbf{A}^H\big)\odot \mathbf{R}_b^{(k)}
\right\|_F^2
+
\mu \|\boldsymbol{\gamma}\|_1.
\label{eq:gamma_step_alt}
\end{aligned}
\end{equation}
Equivalently, by vectorization,~\eqref{eq:gamma_step_alt} can be written as
\begin{equation}
\boldsymbol{\gamma}^{(k+1)}
=
\arg\min_{\boldsymbol{\gamma}\ge \mathbf{0}} 
\left\|\operatorname{vec}(\hat{\mathbf R}_y)-\mathbf\Phi^{(k)}\boldsymbol\gamma\right\|_2^2
+
\mu \|\boldsymbol{\gamma}\|_1,
\end{equation}
where the dictionary matrix $\mathbf{\Phi}^{(k)}\in\mathbb{C}^{M^2\times G}$ is defined as
\begin{equation}
\begin{aligned}
\mathbf{\Phi}^{(k)}
=&
\Big[
\mathrm{vec}\!\big((\mathbf{a}(\vartheta_1)\mathbf{a}^H(\vartheta_1))\odot \mathbf{R}_b^{(k)}\big),
 \\ &\dots,
\mathrm{vec}\!\big((\mathbf{a}(\vartheta_G)\mathbf{a}^H(\vartheta_G))\odot \mathbf{R}_b^{(k)}\big)
\Big].
\end{aligned}
\end{equation}

Step 1b (debiasing refit on the active support):
The $\ell_1$ penalty in~\eqref{eq:gamma_step_alt} introduces shrinkage on the nonzero entries of~$\boldsymbol{\gamma}$. Following the debiasing procedure introduced in~\cite{wright2009sparse}, we refit the coefficients on the support identified by the LASSO solution. Specifically, define the active index set
\begin{equation}
\mathcal{S}^{(k+1)}=\bigl\{\, g\in\{1,\dots,G\} \,\big|\, \gamma_g^{(k+1)}\neq 0 \,\bigr\}.
\label{eq:gamma_support}
\end{equation}
In practice, the condition ($\gamma_g \neq 0$) is implemented using a small numerical threshold. 

Let $\mathbf{\Phi}^{(k)}_{\mathcal{S}}$ denote the submatrix of $\mathbf{\Phi}^{(k)}$ containing the columns indexed by $\mathcal{S}^{(k+1)}$, and let $\boldsymbol{\gamma}_{\mathcal{S}}$ be the corresponding subvector. We then compute the debiased coefficients via the nonnegative least-squares refit
\begin{equation}
\bar{\boldsymbol{\gamma}}_{\mathcal{S}}^{(k+1)}
=
\arg\min_{\boldsymbol{\gamma}_{\mathcal{S}}\ge \mathbf{0}}
\left\|
\mathrm{vec}\!\big(\hat{\mathbf{R}}_{y}\big)
-
\mathbf{\Phi}^{(k)}_{\mathcal{S}}\,\boldsymbol{\gamma}_{\mathcal{S}}
\right\|_2^2.
\label{eq:gamma_debias_nnls}
\end{equation}
Finally, we form a full-length vector $\bar{\boldsymbol{\gamma}}^{(k+1)}\in\mathbb{R}_+^G$ by assigning the refitted values to the indices in $\mathcal S^{(k+1)}$ and setting all other entries to zero. This full-length vector is then used in the subsequent $\mathbf{c}$-update.

Step 2 ($\mathbf{c}$-update):
For the fixed full-length debiased vector $\bar{\boldsymbol{\gamma}}^{(k+1)}$, define
\begin{equation}
\mathbf{R}_x^{(k+1)}
\triangleq
\mathbf{R}_x(\bar{\boldsymbol{\gamma}}^{(k+1)}).
\end{equation}
Using the expansion
\begin{equation}
\mathbf{R}_x^{(k+1)}\odot \mathbf{R}_b(0 ^\circ;\mathbf{c})
=
\sum_{m=0}^{M-1}
c_m\big(\mathbf{R}_x^{(k+1)}\odot \mathbf{T}_m\big),
\end{equation}
the $\mathbf{c}$-update reduces to
\begin{equation}
\begin{aligned}
\mathbf{c}^{(k+1)}
&=
\arg\min_{\mathbf{c}\ge \mathbf{0}} \\&
\left\|
\hat{\mathbf{R}}_y
-
\sum_{m=0}^{M-1}
c_m\big(\mathbf{R}_x^{(k+1)}\odot \mathbf{T}_m\big)
\right\|_F^2
\quad \\&
\text{s.t.}\quad
c_0=1.
\label{eq:c_step_alt}
\end{aligned}
\end{equation}
Equation~\eqref{eq:c_step_alt} can be written in a standard vectorized least-squares form by stacking the basis matrices, similarly to~\eqref{eq:V_def}.

The overall algorithm alternates between the nonnegative LASSO update in~\eqref{eq:gamma_step_alt}, the debiasing refit in~\eqref{eq:gamma_debias_nnls}, and the constrained least-squares update in~\eqref{eq:c_step_alt} until convergence. The alternating LASSO algorithm is summarized in Algorithm 2.

\subsubsection{Joint Optimization Method with LASSO}
The alternating scheme updates $\boldsymbol{\gamma}$ and $\mathbf{c}$ sequentially. 
Alternatively, we use an approximation in which the element-wise inverse of the distortion covariance is applied to the sample covariance.
Let $\mathbf{Q}_b(0 ^\circ)\triangleq \mathbf{R}_b(0 ^\circ)^{\oslash}$ denote the Hadamard (element-wise) inverse of $\mathbf{R}_b(0 ^\circ)$, where the superscript $\oslash$ denotes element-wise reciprocation. 
Since $\mathbf{R}_b(0 ^\circ)$ is real symmetric Toeplitz, $\mathbf{Q}_b(0 ^\circ)$ is also real-valued symmetric Toeplitz and can be parameterized as
\begin{equation}
\mathbf{Q}_b(0 ^\circ;\mathbf{q})
=
\sum_{m=0}^{M-1} q_m\,\mathbf{T}_m,
\label{eq:Rb_inv_toeplitz}
\end{equation}
where $\mathbf{q}=[q_0,\dots,q_{M-1}]^T = [\frac{1}{{c}_0},\dots,\frac{1}{{c}_{M-1}}]^T$ and $\{\mathbf{T}_m\}$ are real-valued Toeplitz basis matrices.

Without additional regularization, the joint estimation may produce oscillatory and physically inconsistent inverse Toeplitz coefficients $\{q_m\}$, which can distort the compensated covariance and lead to inaccurate DoA estimates. 
To stabilize the estimation, we exploit the expected structure of the inverse distortion covariance. 
Since $\mathbf{Q}_b(0^\circ)$ is the element-wise inverse of $\mathbf{R}_b(0^\circ)$, the decay of the distortion correlation with sensor separation implies that $\{q_m\}$ should increase monotonically with the lag index $m$. 
We therefore impose a monotonicity constraint on $\mathbf{q}$ as an additional regularization, which has empirically been found to stabilize the estimation and suppress oscillatory solutions.

We then consider the distortion-compensated covariance mismatch and formulate the joint LASSO-regularized problem
\begin{equation}
\begin{aligned}
(\hat{\boldsymbol{\gamma}},\hat{\mathbf{q}})
&=\operatorname*{arg\,min}_{\boldsymbol{\gamma}\ge \mathbf{0},\,\mathbf{q}\ge \mathbf{0}} \\&
\left\|
\hat{\mathbf{R}}_y \odot \mathbf{Q}_b(0 ^\circ;\mathbf{q})
-
\mathbf{R}_x(\boldsymbol{\gamma})
\right\|_F^2
+
\mu\|\boldsymbol{\gamma}\|_1,
\\
&\text{s.t.}\quad
q_0=1,
q_0 \le q_1 \le \cdots \le q_{M-1},
\label{eq:joint_lasso_obj}
\end{aligned}
\end{equation}
where $\mu>0$ controls the sparsity of $\boldsymbol{\gamma}$. 
The constraint on $q_0$ fixes the global scale.
Problem~\eqref{eq:joint_lasso_obj} is convex in $(\boldsymbol{\gamma},\mathbf{q})$ and can be readily solved using standard convex optimization tools, e.g., CVX in MATLAB~\cite{cvx}.

After solving~\eqref{eq:joint_lasso_obj}, we apply the same debiasing refit on the active support as in the alternating method, but here we jointly refit both $\boldsymbol{\gamma}$ and $\mathbf{q}$. 
Specifically, we identify the active support of the LASSO solution and then resolve the joint problem with the $\ell_1$ term removed, i.e., using a pure $\ell_2$  objective while enforcing the same structural constraints on $\mathbf{q}$.
Let
\begin{equation}
\mathcal{S}
=
\bigl\{\, g\in\{1,\dots,G\}\,\big|\,\hat{\gamma}_g\neq 0 \,\bigr\},
\label{eq:joint_support}
\end{equation}
denote the active index set. In the refit, we optimize only the active subvector $\boldsymbol{\gamma}_{\mathcal{S}}$, whose entries are nonzero:
\begin{equation}
\begin{aligned}
(\hat{\boldsymbol{\gamma}}_{\mathcal{S}},\hat{\mathbf{q}})
&=
\operatorname*{arg\,min}_{\boldsymbol{\gamma}_{\mathcal{S}}\ge \mathbf{0},\,\mathbf{q}\ge \mathbf{0}}
\left\|
\hat{\mathbf{R}}_y \odot \mathbf{Q}_b(0 ^\circ;\mathbf{q})
-
\mathbf{R}_x(\boldsymbol{\gamma}_{\mathcal{S}})
\right\|_F^2, \\ 
&\text{s.t.}\quad
q_0=1,
q_0 \le q_1 \le \cdots \le q_{M-1},
\label{eq:joint_lasso_refit}
\end{aligned}
\end{equation}
Note that this refit step keeps the estimated support $\mathcal{S}$ fixed and therefore does not change the DoA estimates implied by the active grid locations. Its role is to obtain less biased estimates of the corresponding source powers (encoded in $\hat{\boldsymbol{\gamma}}_{\mathcal{S}}$) and of the inverse-distortion parameters $\hat{\mathbf{q}}$. In other words, the refit is optional if the goal is DoA localization only, but it is beneficial when accurate power and distortion estimates are also of interest.

The joint LASSO algorithm is summarized in Algorithm 3.

\begin{algorithm}[htb]
\caption{Multiple-Source Calibration via Alternating LASSO}
\begin{algorithmic}[1]
\STATE \textbf{Input}: Sample covariance $\hat{\mathbf{R}}_y$, angular grid $\{\vartheta_g\}_{g=1}^{G}$, array manifold
\[
\mathbf{A}=\big[\mathbf{a}(\vartheta_1),\dots,\mathbf{a}(\vartheta_G)\big],
\]
Toeplitz basis matrices $\{\mathbf{T}_m\}_{m=0}^{M-1}$, regularization parameter $\mu$.
\STATE \textbf{Initialize}: $\mathbf{c}^{(0)} \ge 0$ with $c_0^{(0)} = 1$.
\FOR{$k=0,1,\dots$ until convergence}
    \STATE Form
    \[
    \mathbf{R}_b(0 ^\circ;\mathbf{c}^{(k)})
    =
    \sum_{m=0}^{M-1} c_m^{(k)} \mathbf{T}_m .
    \]
    \STATE Update the sparse angular power vector by solving
\[
\begin{aligned}
\hat {\boldsymbol{\gamma}}^{(k+1)}
&=
\arg\min_{\boldsymbol{\gamma}\ge 0}
\;
\left\|
\hat{\mathbf{R}}_y
-
\mathbf{R}_x(\boldsymbol{\gamma})\odot
\mathbf{R}_b(0 ^\circ;\mathbf{c}^{(k)})
\right\|_F^2
\\
&\quad
+\mu \|\boldsymbol{\gamma}\|_1 .
\end{aligned}
\]
    where
    \[
    \mathbf{R}_x(\boldsymbol{\gamma})
    =
    \mathbf{A}\,\mathrm{diag}(\boldsymbol{\gamma})\,\mathbf{A}^H .
    \]
    \STATE Debias the estimate by removing the $\ell_1$ term and resolving the same problem on the active support of ${\hat {\boldsymbol{\gamma}}}^{(k+1)}$.
    \STATE Update the distortion coefficients by solving
    \[
    \mathbf{c}^{(k+1)}
    =
    \arg\min_{\mathbf{c}\ge 0}
    \left\|
    \hat{\mathbf{R}}_y
    -
    \mathbf{R}_x(\boldsymbol{\gamma}^{(k+1)})\odot
    \mathbf{R}_b(0 ^\circ;\mathbf{c})
    \right\|_F^2
    \]
    subject to
    \[
    c_0 = 1.
    \]
\ENDFOR
\STATE \textbf{Output}: $\hat{\boldsymbol{\gamma}}$, $\hat{\mathbf{c}}$, 
\[
\hat{\mathbf{R}}_b=\mathbf{R}_b(0 ^\circ;\hat{\mathbf{c}}),
\qquad
\hat{\mathbf{R}}_x=\mathbf{R}_x(\hat{\boldsymbol{\gamma}}).
\]
The estimated DoAs are given by the nonzero entries of~$\hat{\boldsymbol{\gamma}}$.
\end{algorithmic}
\label{alg:alt_lasso_multi}
\end{algorithm}

\begin{algorithm}[htb]
\caption{Multiple-Source Calibration via Joint LASSO}
\begin{algorithmic}[1]
\STATE \textbf{Input}: Sample covariance $\hat{\mathbf{R}}_y$, angular grid $\{\vartheta_g\}_{g=1}^{G}$, array manifold
\[
\mathbf{A}=\big[\mathbf{a}(\vartheta_1),\dots,\mathbf{a}(\vartheta_G)\big],
\]
Toeplitz basis matrices $\{\mathbf{T}_m\}_{m=0}^{M-1}$, regularization parameter $\mu$.
\STATE Parameterize the inverse distortion covariance as
\[
\mathbf{Q}_b(0 ^\circ;\mathbf{q})
=
\sum_{m=0}^{M-1} q_m \mathbf{T}_m .
\]
\STATE Parameterize the undistorted covariance on the angular grid:
\[
\mathbf{R}_x(\boldsymbol{\gamma})
=
\mathbf{A}\,\mathrm{diag}(\boldsymbol{\gamma})\,\mathbf{A}^H .
\]
\STATE Jointly estimate $(\boldsymbol{\gamma},\mathbf{q})$ by solving
\[
(\hat{\boldsymbol{\gamma}},\hat{\mathbf{q}})
=
\operatorname*{arg\,min}_{\boldsymbol{\gamma}\ge 0,\;\mathbf{q}\ge 0}
\left\|
\hat{\mathbf{R}}_y \odot \mathbf{Q}_b(0 ^\circ;\mathbf{q})
-
\mathbf{R}_x(\boldsymbol{\gamma})
\right\|_F^2
+
\mu\|\boldsymbol{\gamma}\|_1
\]
subject to
\[
q_0=1,
\qquad
q_0 \le q_1 \le \cdots \le q_{M-1}.
\]
\STATE \textbf{Debiasing refit (optional):} Debias the estimate by removing the $\ell_1$ term and resolving the same problem on the active support of $\hat{\boldsymbol{\gamma}}$.
\STATE \textbf{Output}: $\hat{\boldsymbol{\gamma}}$, $\hat{\mathbf{q}}$,
\[
\hat{\mathbf{R}}_x=\mathbf{R}_x(\hat{\boldsymbol{\gamma}}),
\qquad
\hat{\mathbf{Q}}_b=\mathbf{Q}_b(0 ^\circ;\hat{\mathbf{q}}).
\]
The estimated DoAs are given by the nonzero entries of~$\hat{\boldsymbol{\gamma}}$.
\end{algorithmic}
\label{alg:joint_lasso_multi}
\end{algorithm}

\subsubsection{Joint Optimization Method with Nuclear Norm}
In the joint LASSO formulation, $\mathbf{R}_x$ is parameterized on an angular grid via $\boldsymbol{\gamma}$. 
Alternatively, we may relax the undistorted signal covariance $\mathbf{R}_x$ as an HT covariance matrix. 
To avoid ambiguity while capturing the fact that a small number of sources implies a low-rank signal covariance, we promote low rank by adding a nuclear-norm penalty~\cite{recht2010guaranteed}.

As in the previous joint formulation, we parameterize the Hadamard inverse distortion covariance by a real-valued Toeplitz expansion \eqref{eq:Rb_inv_toeplitz}. We then solve the nuclear-norm-regularized joint estimation problem
\begin{equation}
\begin{aligned}
(\hat{\mathbf{R}}_x,\hat{\mathbf{q}})
&=\operatorname*{arg\,min}_{\mathbf{R}_x \succeq \mathbf{0},\,\mathbf{q}\ge \mathbf{0}} \\&
\left\|
\hat{\mathbf{R}}_y \odot \mathbf{Q}_b(0 ^\circ;\mathbf{q})
-
\mathbf{R}_x
\right\|_F^2
+
\epsilon\|\mathbf{R}_x\|_*,
\\
&\text{s.t.}\quad
q_0=1,
q_0 \le q_1 \le \cdots \le q_{M-1}, \\
&\text{and  $\mathbf{R}_x =  \sum_{m=0}^{2M-2} \eta_m \boldsymbol{\Sigma}_m,\; \eta_m\in\mathbb{R}.$}
\label{eq:joint_nuclear_obj}
\end{aligned}
\end{equation}
where $\epsilon>0$ is the weighting coefficient of the nuclear norm $\|\cdot\|_*$. \eqref{eq:joint_nuclear_obj} is convex and can be readily solved using standard convex optimization tools, e.g., CVX in MATLAB.

As before, after obtaining $(\hat{\mathbf{R}}_x,\hat{\mathbf{q}})$ from~\eqref{eq:joint_nuclear_obj}, we optionally perform a debiasing refit to mitigate the shrinkage induced by the nuclear-norm penalty. 
Specifically, we compute the eigendecomposition of $\hat{\mathbf{R}}_x$ and retain the $K$ dominant eigenvectors, denoted by $\hat{\mathbf{U}}_s\in\mathbb{C}^{M\times K}$. 
In the refit, the covariance matrix is restricted to the estimated signal subspace and parameterized as $\mathbf{R}_{x}^{\mathrm{sub}}=\hat{\mathbf{U}}_s\mathbf{Z}\hat{\mathbf{U}}_s^H$ with $\mathbf{Z}\succeq 0$.
We then remove the nuclear-norm penalty and re-optimize $\mathbf{Z}$ and $\mathbf{q}$ using a pure $\ell_2$ data-fitting criterion, subject to the same structural constraints on $\mathbf{q}$. 
Finally, MUSIC is applied to the refitted covariance estimate
\[
\hat{\mathbf{R}}_x^{\mathrm{refit}}
=
\hat{\mathbf{U}}_s\,\hat{\mathbf{Z}}\,\hat{\mathbf{U}}_s^H
\]
to extract the DoAs. 
The nuclear-norm-based joint estimation is summarized in Algorithm~4.

\begin{algorithm}[htb]
\caption{Multiple-Source Calibration via Joint Nuclear-Norm Regularization}
\begin{algorithmic}[1]
\STATE \textbf{Input}: Sample covariance $\hat{\mathbf{R}}_y$, Toeplitz basis matrices $\{\mathbf{T}_m\}_{m=0}^{M-1}$, regularization parameter $\epsilon$.
\STATE Parameterize the inverse distortion covariance as
\[
\mathbf{Q}_b(0 ^\circ;\mathbf{q})
=
\sum_{m=0}^{M-1} q_m \mathbf{T}_m .
\]
\STATE Jointly estimate $(\mathbf{R}_x,\mathbf{q})$ by solving
\[
\begin{aligned}
(\hat{\mathbf{R}}_x,\hat{\mathbf{q}})
&=
\operatorname*{arg\,min}_{\mathbf{R}_x \succeq 0,\;\mathbf{q}\ge 0}
\;
\left\|
\hat{\mathbf{R}}_y \odot \mathbf{Q}_b(0 ^\circ;\mathbf{q})
-
\mathbf{R}_x
\right\|_F^2
\\
&\quad
+
\epsilon \|\mathbf{R}_x\|_* .
\end{aligned}
\]
subject to
\[
q_0=1,
\qquad
q_0 \le q_1 \le \cdots \le q_{M-1},
\]
\text{and  $\mathbf{R}_x =  \sum_{m=0}^{2M-2} \eta_m \boldsymbol{\Sigma}_m,\; \eta_m\in\mathbb{R}.$}
\STATE \textbf{Debiasing refit (optional):} Compute the $K$ dominant eigenvectors of $\hat{\mathbf{R}}_x$ to form $\hat{\mathbf{U}}_s$, restrict $\mathbf{R}_x=\hat{\mathbf{U}}_s\mathbf{Z}\hat{\mathbf{U}}_s^{H}$ with $\mathbf{Z}\succeq 0$, and re-solve the problem with the nuclear-norm term removed to obtain $(\hat{\mathbf{Z}},\hat{\mathbf{q}})$ under the same constraints on $\mathbf{q}$.
\STATE Apply MUSIC to the (refitted) covariance to obtain the DoA estimates.
\STATE \textbf{Output}: $\hat{\mathbf{R}}_x$ (or $\hat{\mathbf{R}}_x^{\mathrm{refit}}$), $\hat{\mathbf{q}}$, and the estimated DoAs.
\end{algorithmic}
\label{alg:nuclear_multi}
\end{algorithm}

%% file: performance_bounds.tex
\section{Performance Bounds and Identifiability}
\label{sec:performance_bounds}
This section gives single-source performance bounds for the two covariance models considered in this paper. The first bound corresponds to the relaxed structured-covariance model used by the single-source calibration method in Sec.~\ref{sec:calibration}, and is referred to as the structured-covariance CRLB (SC--CRLB). The second bound uses the physics-informed distortion covariance in~\eqref{eq:Rb_piecewise}, together with the empirical correlation law in~\eqref{eq:empirical_alpha} and the angle-dependent separation in~\eqref{eq:d_theta}; it is referred to as the physics-informed CRLB (PI--CRLB).

The bounds are restricted to the single-source case, since in  the original multi-source model~\eqref{eq:covariance_measurement}, assigning an independent angle-dependent distortion covariance $\mathbf{R}_b(\theta_n)$ to each source leads to the non-identifiability discussed in Prop.~\ref{prop:multi_identifiability}.

We use the standard Slepian--Bangs formula for zero-mean circular complex Gaussian observations~\cite[App.~B.3]{stoica2005spectral}. For a real parameter vector $\mathbf{g}$ and covariance $\mathbf{R}_y(\mathbf{g})$, define
\begin{equation}
\mathbf{s}_{g_i}
=\mathrm{vec}\!\left(\frac{\partial\mathbf{R}_y}{\partial g_i}\right),
\qquad
\mathbf{Q}=\mathbf{R}_y^{-T}\otimes\mathbf{R}_y^{-1}.
\label{eq:sensitivity_and_Q_def}
\end{equation}
The FIM entries are
\begin{equation}
[\mathbf{F}(\mathbf{g})]_{ij}
=
T\,\mathrm{Re}\!\left\{
\mathbf{s}_{g_i}^{H}\mathbf{Q}\mathbf{s}_{g_j}
\right\}.
\label{eq:FIM_SB}
\end{equation}
For a partition $\mathbf{g}=[\boldsymbol{\psi}^{T},\boldsymbol{\nu}^{T}]^{T}$ into parameters of interest $\boldsymbol{\psi}$ and nuisance parameters $\boldsymbol{\nu}$, the nuisance-eliminated information matrix is obtained by the Schur complement,
\begin{equation}
\mathbf{F}_{\mathrm{eff}}
=
\mathbf{F}_{\psi\psi}
-
\mathbf{F}_{\psi\nu}
\mathbf{F}_{\nu\nu}^{-1}
\mathbf{F}_{\nu\psi},
\label{eq:Feff_schur}
\end{equation}
and the CRLB is $\mathbf{F}_{\mathrm{eff}}^{-1}$ when $\mathbf{F}_{\mathrm{eff}}$ is nonsingular.

\subsection{Structured-Covariance CRLB}
\label{subsec:SC_CRLB}
The SC--CRLB is based on the relaxed single-source covariance model
\begin{equation}
\mathbf{R}_y^{\mathrm{SC}}
=
\mathbf{S}(\theta)\odot\mathbf{R}_{\mathrm{SC}}(\mathbf{c}),
\qquad
\mathbf{R}_{\mathrm{SC}}(\mathbf{c})
=
\sum_{k=0}^{M-1}c_k\mathbf{T}_k,
\label{eq:Ry_SC_model}
\end{equation}
where $\mathbf{S}(\theta)=\mathbf{a}(\theta)\mathbf{a}^{H}(\theta)$ is defined in~\eqref{eq:S_def}, and $\mathbf{R}_{\mathrm{SC}}(\mathbf{c})$ is a real symmetric Toeplitz matrix. The parameter of interest is $\theta$, while $\mathbf{c}=[c_0,\ldots,c_{M-1}]^T$ is treated as nuisance.

\begin{lemma}
\label{lemma1}
For the SC model in~\eqref{eq:Ry_SC_model}, let the FIM of $[\theta,\mathbf{c}^{T}]^{T}$ be computed from~\eqref{eq:FIM_SB} and partitioned as
\begin{equation}
\mathbf{F}_{\mathrm{SC}}
=
\begin{bmatrix}
F_{\theta\theta} & \mathbf{F}_{\theta c}\\
\mathbf{F}_{c\theta} & \mathbf{F}_{cc}
\end{bmatrix}.
\end{equation}
Then the SC--CRLB for $\theta$ is
\begin{equation}
\mathrm{CRLB}_{\mathrm{SC}}(\theta)
=
\left(
F_{\theta\theta}
-
\mathbf{F}_{\theta c}
\mathbf{F}_{cc}^{-1}
\mathbf{F}_{c\theta}
\right)^{-1}.
\label{eq:CRLB_theta_SC}
\end{equation}
\end{lemma}
The covariance sensitivities needed for Lemma~\ref{lemma1} are listed in Appendix~\ref{App.A}. 

\subsection{Physics-Informed CRLB}
\label{subsec:PI_CRLB}
For the PI--CRLB, the distortion covariance is constrained by the physical model in Sec.~\ref{sec:2}. In order to obtain an identifiable parameterization, we define
\begin{equation}
\kappa
\triangleq
\frac{a_1R}{a_2R+1}.
\label{eq:kappa_def}
\end{equation}
With this reparameterization, write the physics-informed distortion covariance as $\mathbf{R}_b(\theta,\kappa,a_3)$, obtained from~\eqref{eq:Rb_piecewise} by evaluating~\eqref{eq:empirical_alpha} with the separation in~\eqref{eq:d_theta}. The corresponding single-source covariance is
\begin{equation}
\mathbf{R}_y^{\mathrm{PI}}
=
\rho\,
\big[
\mathbf{S}(\theta)\odot\mathbf{R}_b(\theta,\kappa,a_3)
\big]
+
\sigma_n^2\mathbf{I}_M,
\label{eq:Ry_PI_model}
\end{equation}
where $\rho\triangleq |\beta|^2\sigma_s^2$ is the effective received source power. The parameters of interest and nuisance parameters are
\begin{equation}
\boldsymbol{\psi}_{\mathrm{PI}}
=
[\theta,\kappa,a_3]^T,
\qquad
\boldsymbol{\nu}_{\mathrm{PI}}
=
[\rho,\sigma_n^2]^T.
\label{eq:param_PI_interest_nuis}
\end{equation}

\begin{lemma}
\label{lemma2}
For the PI model in~\eqref{eq:Ry_PI_model}, let the FIM of $[\boldsymbol{\psi}_{\mathrm{PI}}^T,\boldsymbol{\nu}_{\mathrm{PI}}^T]^T$ be computed from~\eqref{eq:FIM_SB}. Partition it according to~\eqref{eq:Feff_schur}. Then the PI--CRLB for $[\theta,\kappa,a_3]^T$ is
\begin{equation}
\mathbf{C}_{\mathrm{PI}}
=
\left(
\mathbf{F}_{\psi\psi}
-
\mathbf{F}_{\psi\nu}
\mathbf{F}_{\nu\nu}^{-1}
\mathbf{F}_{\nu\psi}
\right)^{-1}.
\label{eq:CRLB_PI_matrix}
\end{equation}
The marginal bounds are
\begin{align}
\mathrm{CRLB}_{\mathrm{PI}}(\theta)&=[\mathbf{C}_{\mathrm{PI}}]_{1,1},
\label{eq:CRLB_PI_theta}\\
\mathrm{CRLB}_{\mathrm{PI}}(\kappa)&=[\mathbf{C}_{\mathrm{PI}}]_{2,2},
\label{eq:CRLB_PI_kappa}\\
\mathrm{CRLB}_{\mathrm{PI}}(a_3)&=[\mathbf{C}_{\mathrm{PI}}]_{3,3}.
\label{eq:CRLB_PI_a3}
\end{align}
\end{lemma}
The covariance sensitivities needed for Lemma~\ref{lemma2} are listed in Appendix~\ref{App.B}. 

Figure~\ref{crlb_vs_rmse} compares the SC--CRLB and PI--CRLB for $\theta$ with the RMSE of the proposed single-source estimators, while Fig.~\ref{fig:crlb_beta_a3} reports the corresponding bounds for $\kappa$ and $a_3$. These results indicate that the empirical parameters are locally identifiable under the considered single-source operating conditions. However, converting $(\kappa,a_3)$ into physical quantities such as rain rate and propagation range still requires a calibrated mapping, for example through the tabulated coefficients in~\cite[Table~II]{yektakhah2024model}.

%% file: numerical_results.tex
\section{Numerical Results and Performance Analysis}
\label{Numerical Results}

In this section, we present simulation results to evaluate the performance of the proposed methods and compare them with the corresponding Cramér--Rao lower bounds (CRLBs), which are available in the single-source case.

\subsection{Single-Source Case}
We first consider the single-source scenario to validate the proposed calibration method. A ULA with $M = 8$ elements and half-wavelength spacing is employed, observing a single source located at $\theta = 30^\circ$ with a total of $T = 1000$ snapshots. For the rain-induced distortion model, a rain rate of $50\,\mathrm{mm/hr}$ and a propagation range of $200\,\mathrm{m}$ are assumed.

\begin{figure}[t]
    \centering
    \includegraphics[width=\columnwidth]{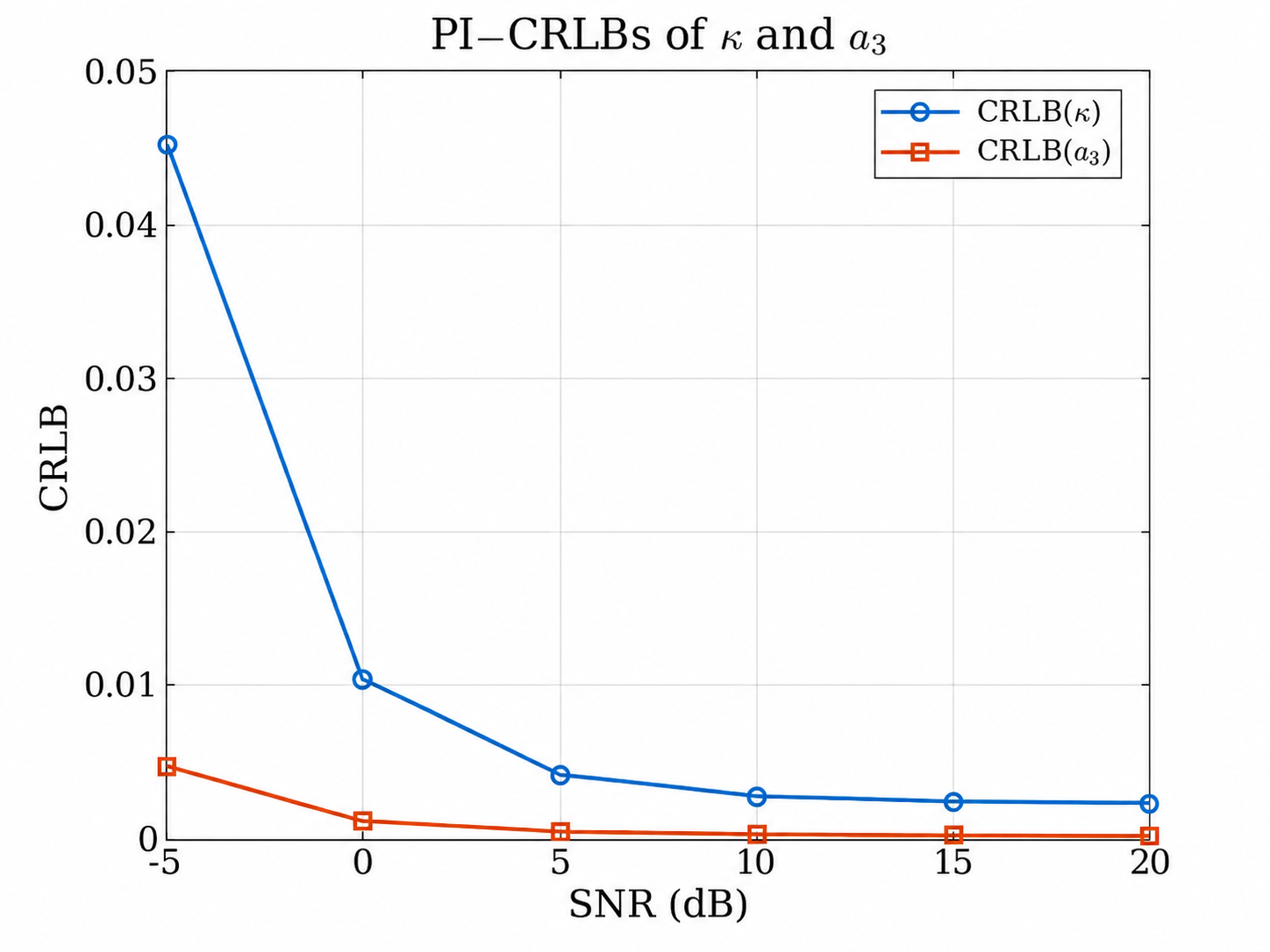}
    \caption{PI--CRLBs of the empirical parameters $\kappa$ and $a_3$ versus SNR in the single-source physics-informed model.}
    \label{fig:crlb_beta_a3}
\end{figure}

Figure~\ref{distortioncov} compares the subdiagonal values of the estimated distortion covariance matrix with the true distortion covariance matrix at $\mathrm{SNR}=20\ \mathrm{dB}$. The strong agreement across all subdiagonals confirms that the proposed estimation procedure accurately captures the distortion statistics.

\begin{figure}[htbp]
    \centering
    \includegraphics[width=\columnwidth]{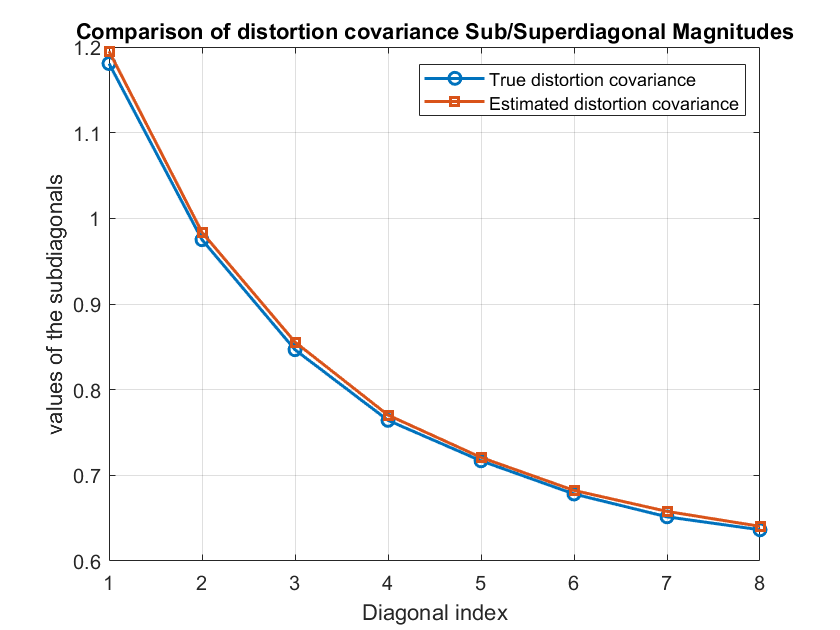}
    \caption{Subdiagonal-wise comparison between the estimated and true distortion covariance matrices at $\mathrm{SNR}=20\ \mathrm{dB}$ and for $T = 1000$ snapshots.}
    \label{distortioncov}
\end{figure}

Figure~\ref{bartlett} compares the Bartlett (conventional beamforming) spatial spectra for three cases: an undistorted array (clean reference), a rain-distorted array without calibration, and a rain-distorted array after calibration using either LS or WLS. 
We adopt Bartlett beamforming rather than MUSIC in this comparison because, in the single-source undistorted case, Bartlett is equivalent to the maximum-likelihood (ML) estimator under spatially white Gaussian noise.

\begin{figure}[htbp]
    \centering
    \includegraphics[width=\columnwidth]{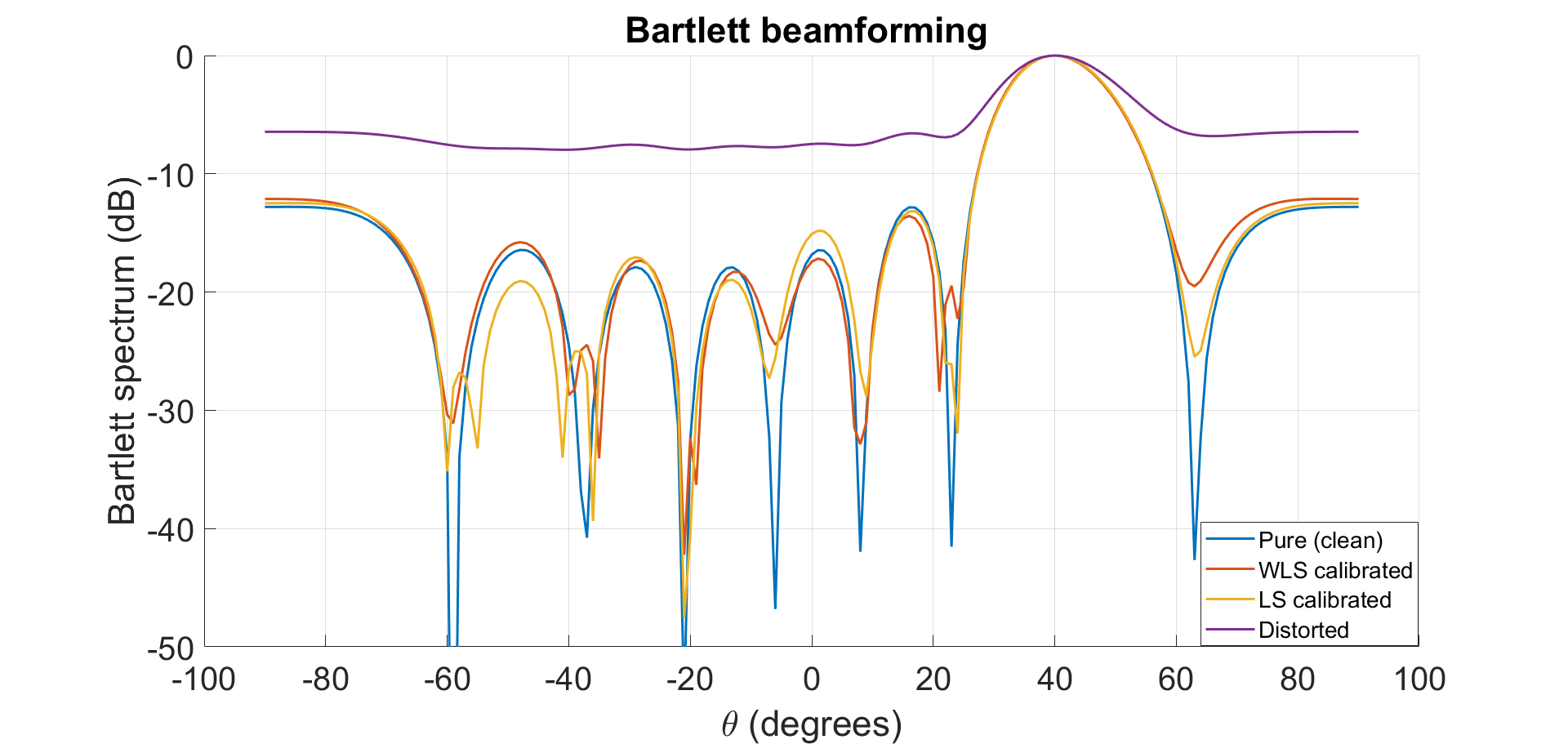}
    \caption{Bartlett spatial spectra under three conditions: clean reference, rain distortion without calibration, and rain distortion with LS/WLS calibration. }
    \label{bartlett}
\end{figure}

Without calibration, rain-induced distortion raises the spectral floor and fills in the nulls, resulting in a broadened mainlobe and reduced peak-to-sidelobe contrast, which in turn degrades angular resolution. 
Both the LS and WLS based calibration procedures substantially restore the clean beampattern, producing a pronounced peak at the true direction and improved sidelobe suppression. 
Moreover, WLS more closely matches the clean reference, yielding a slightly sharper peak and deeper nulls than LS, which highlights the benefit of statistically informed weighting.

Figure~\ref{rmse_single} reports the root mean squared error (RMSE) of the DoA estimates versus SNR. Both the proposed calibration-enhanced Bartlett beamformer and the conventional Bartlett baseline are evaluated. The results clearly show that calibration substantially improves estimation accuracy over the uncalibrated case, and that the WLS-based variant consistently outperforms its LS counterpart.

Furthermore, Fig.~\ref{crlb_vs_rmse} compares the RMSE of the proposed methods with both the model-based and the physics-based CRLBs. The results show that the proposed method closely approaches the theoretical limits at sufficiently high SNR, with the WLS-based variant lying even closer to the corresponding CRLB than its LS counterpart. 
As expected, the physics-based CRLB provides the lowest variance bound, while the model-based CRLB remains very close to it.

\begin{figure}[htbp]
    \centering    \includegraphics[width=\columnwidth]{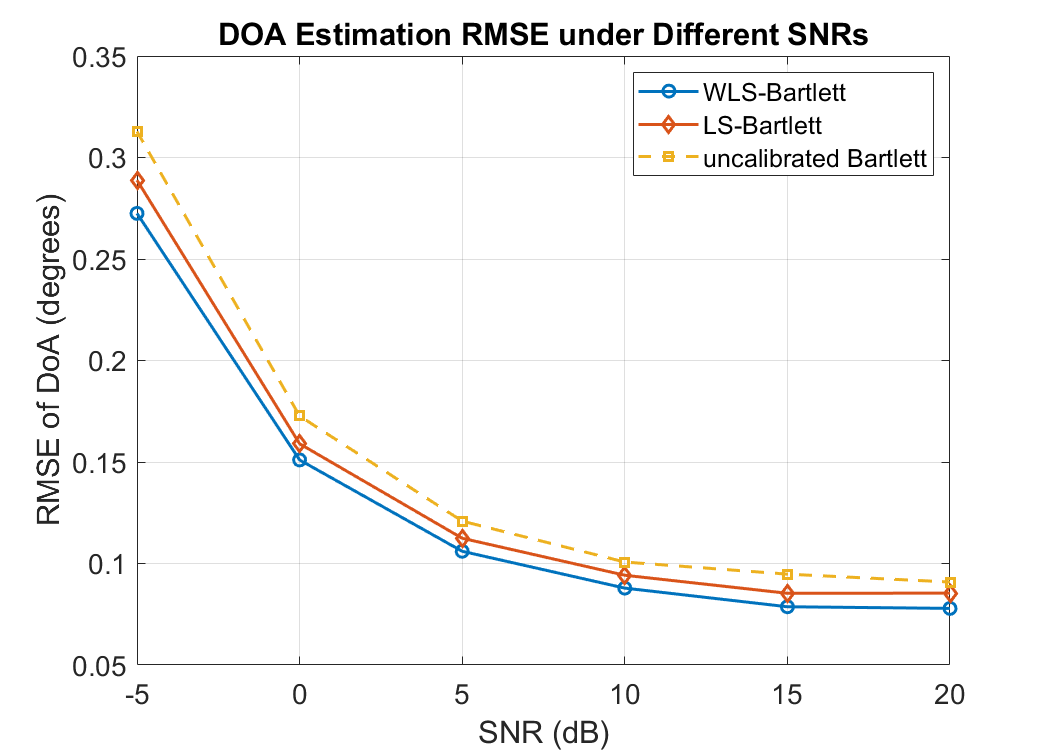}
    \caption{RMSE comparison between the proposed calibration-enhanced Bartlett and Bartlett without calibration, with $T =1000$ snapshots.}
    \label{rmse_single}
\end{figure}

\begin{figure}[htbp]
    \centering
    \includegraphics[width=\columnwidth]{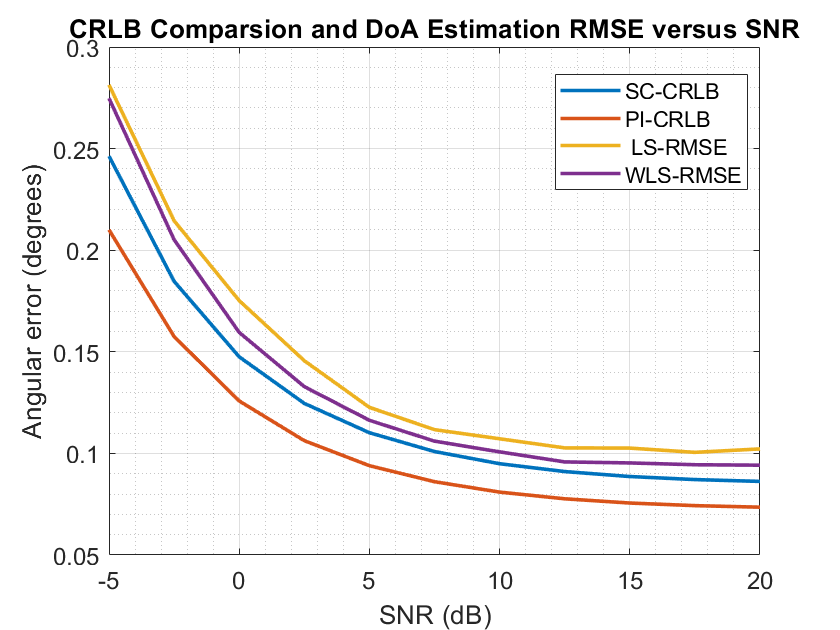}
    \caption{Comparison of RMSE with the model-based CRLB and the physics-based CRLB in the single-source case.}
    \label{crlb_vs_rmse}
\end{figure}

\subsection{Multiple-Source Case}
We next evaluate the proposed multi-source calibration methods using a representative two-source experiment, where the sources are located at ($10^\circ$) and ($20^\circ$). The second source is set to be $10\,\mathrm{dB}$ weaker than the first. 
All other parameters are kept identical to those in the single-source case.

The resulting spectra for $T=1000$ snapshots are shown in Fig.~\ref{fig:spectrum_multiple}. 
The top panel compares the MUSIC spectrum computed from the rain-distorted sample covariance with that obtained after calibration using the proposed joint nuclear-norm formulation. 
Conventional MUSIC applied directly to the distorted covariance fails to resolve the two sources, whereas the calibrated covariance yields two distinct peaks at the correct angles.

The bottom panel compares the alternating LASSO and joint LASSO approaches. 
Both methods localize the DoAs accurately, but their reconstructed source powers differ. 
The alternating LASSO estimates the source powers as $1.03$ and $0.10$, corresponding to a linear ratio of approximately $10.3$ (about $10.1\,\mathrm{dB}$), which closely matches the prescribed $10\,\mathrm{dB}$ difference. 
The joint LASSO yields $0.91$ and $0.11$, corresponding to a ratio of approximately $8.27$ (about $9.2\,\mathrm{dB}$). 
Thus, the alternating formulation provides a more accurate power reconstruction in this experiment, because the joint formulation is more sensitive to mismatch in the estimated distortion covariance due to its element-wise division operation.

\begin{figure}[t]
  \centering
  \includegraphics[width=\columnwidth]{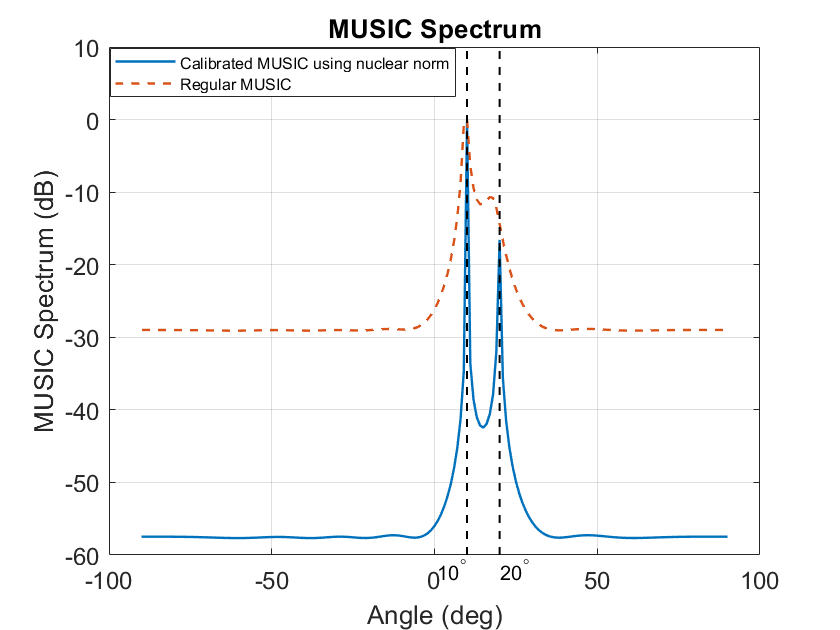}

  \vspace{0.35em}

  \includegraphics[width=\columnwidth]{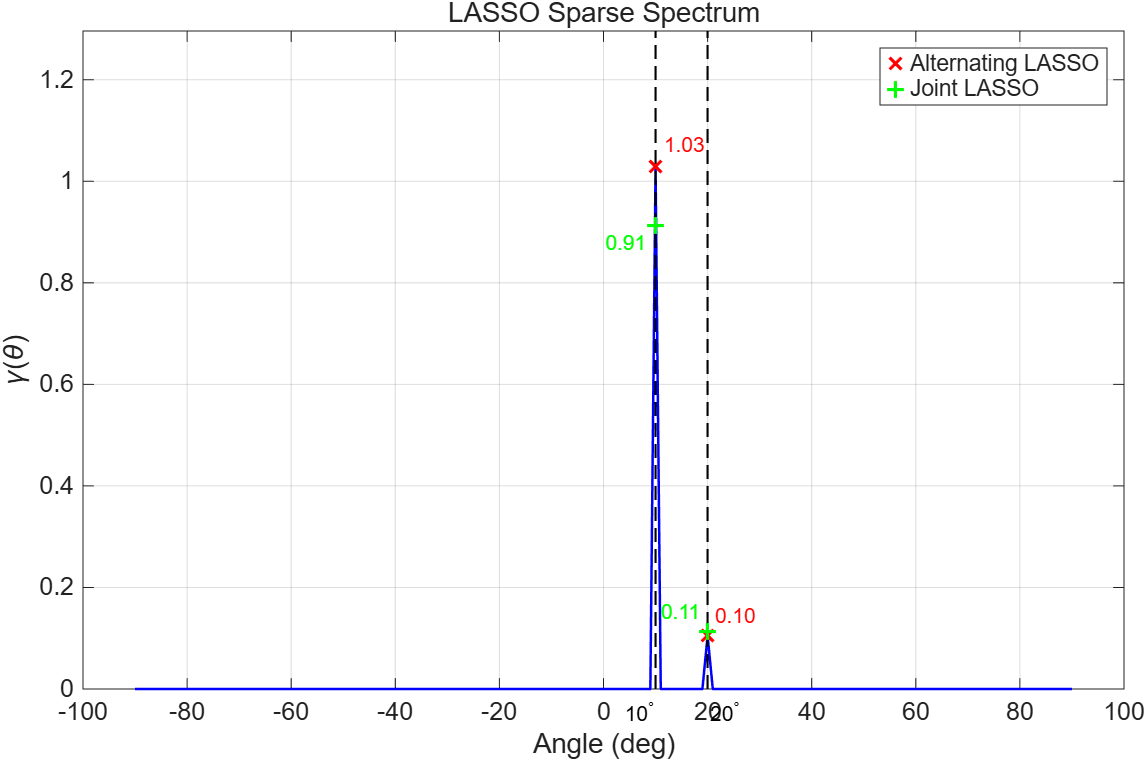}
  \caption{Two-source case under heavy rain ($\theta=\{10^\circ,20^\circ\}$, $T=1000$ snapshots) with a $10\,\mathrm{dB}$ power difference between the sources. 
  (Top) MUSIC spectrum computed from the rain-distorted sample covariance and after calibration using the proposed joint nuclear-norm formulation. 
  (Bottom) Recovered angular power spectra from the alternating LASSO and the joint LASSO formulations.}
  \label{fig:spectrum_multiple} 
\end{figure}

\begin{figure}[htbp]
    \centering
    \includegraphics[width=\columnwidth]{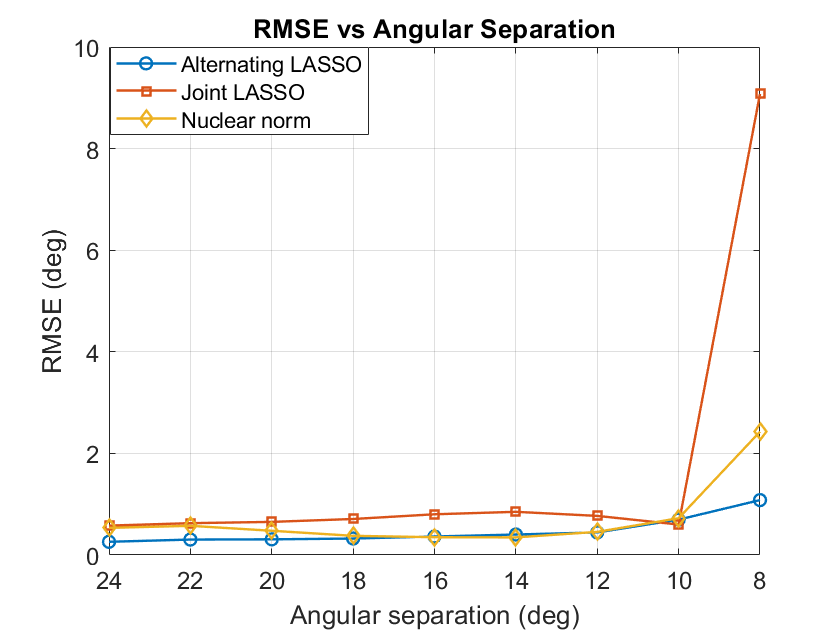}
    \caption{RMSE versus angular separation for two sources symmetrically placed around $10^\circ$, as the spacing decreases from $24^\circ$ to $8^\circ$, using $T=1000$ snapshots.}
    \label{fig:separation}
\end{figure}

To evaluate the capability of the proposed approaches in resolving two closely spaced sources, we computed the RMSE as the angular separation between two sources symmetrically placed around $10^\circ$ was reduced from $24^\circ$ to $8^\circ$ under heavy rain ($50~\mathrm{mm/hr}$). The results in Fig.~\ref{fig:separation} indicate that the alternating LASSO method consistently achieves the lowest RMSE over the entire range of separations and remains relatively stable even as the sources become closely spaced. For moderate and large separations, all three methods provide comparable accuracy, with RMSE values generally below $1^\circ$. However, as the separation decreases, clear differences emerge. In particular, at $8^\circ$, the joint LASSO method exhibits a dramatic performance degradation, with the RMSE increasing sharply to about $9^\circ$, while the nuclear-norm-based method also deteriorates noticeably to around $2.5^\circ$. By contrast, the alternating LASSO approach degrades much more mildly and still maintains an RMSE close to $1^\circ$. This behavior is consistent with the fact that the joint LASSO formulation relies on an element-wise division (or equivalently, Hadamard inversion) of the distortion covariance, which can amplify modeling mismatch and estimation noise.

%% file: conclusion.tex
\section{Conclusion}
This work studied DoA estimation under weather-induced phase and amplitude distortions. Based on a physically motivated rain-propagation model, we developed a structured covariance formulation that preserves the Toeplitz array structure while capturing distortion statistics. The proposed covariance-matching estimators enable calibration-enhanced DoA recovery by separating distortion effects from the signal covariance. We further showed that the general multi-source formulation is non-identifiable, introduced a reduced approximation for tractable estimation, and derived Cramér--Rao lower bounds from both structured-covariance and physics-informed perspectives.

The numerical results confirm the effectiveness of the proposed framework. In the single-source case, calibration substantially improves angular accuracy compared with Bartlett beamforming on distorted data. In multi-source scenarios, conventional subspace processing loses resolvability under rain-induced distortions, especially for closely spaced sources. The proposed structured-covariance methods recover distinct peaks and reduce estimation error, with the alternating LASSO method showing the most robust performance as angular separation decreases.

Future work should address the remaining non-identifiability in general multi-source formulations, extend the analysis beyond narrowband far-field ULAs, and quantify robustness to model mismatch, calibration errors, and limited snapshots. Validation with measured adverse-weather radar data and integration into radar--LiDAR--camera fusion pipelines are also important steps toward practical deployment.

%% file: appendix.tex
\appendices

\section{Covariance Sensitivities for Lemma~\ref{lemma1}}
\label{App.A}
The general CRLB calculation follows directly from the Slepian--Bangs formula in~\eqref{eq:FIM_SB} and the Schur complement in~\eqref{eq:Feff_schur}. We only list the covariance sensitivities used to form the FIM.

Let
\begin{equation}
\mathbf{D}=\mathrm{diag}(0,1,\ldots,M-1).
\label{eq:D_def_AppA}
\end{equation}
For $\mathbf{S}(\theta)=\mathbf{a}(\theta)\mathbf{a}^{H}(\theta)$,
\begin{equation}
\frac{\partial\mathbf{S}(\theta)}{\partial\theta}
=
j2\pi d_0\cos\theta
\big(\mathbf{D}\mathbf{S}(\theta)-\mathbf{S}(\theta)\mathbf{D}\big).
\label{eq:dS_dtheta_AppA}
\end{equation}
Therefore, for the SC model in~\eqref{eq:Ry_SC_model},
\begin{equation}
\frac{\partial\mathbf{R}_y^{\mathrm{SC}}}{\partial\theta}
=
\frac{\partial\mathbf{S}(\theta)}{\partial\theta}
\odot
\mathbf{R}_{\mathrm{SC}}(\mathbf{c}),
\label{eq:dRy_dtheta_SC_AppA}
\end{equation}
\begin{equation}
\frac{\partial\mathbf{R}_y^{\mathrm{SC}}}{\partial c_k}
=
\mathbf{S}(\theta)\odot\mathbf{T}_k,
\qquad
k=0,\ldots,M-1.
\label{eq:dRy_dck_SC_AppA}
\end{equation}
Substituting the vectorized forms of~\eqref{eq:dRy_dtheta_SC_AppA} and~\eqref{eq:dRy_dck_SC_AppA} into~\eqref{eq:FIM_SB}, followed by the Schur complement over $\mathbf{c}$, gives~\eqref{eq:CRLB_theta_SC}.

\section{Covariance Sensitivities for Lemma~\ref{lemma2}}
\label{App.B}
For the PI model, write $k=|m-\ell|$ and define, for $k\geq 1$,
\begin{align}
d_k(\theta)&=kd_0\cos\theta,
\label{eq:dk_theta_AppB}\\
\zeta_k(\theta,a_3)&=\frac{d_k(\theta)}{a_3d_k(\theta)+1},
\label{eq:zeta_AppB}\\
\alpha_k&=e^{-\kappa\zeta_k(\theta,a_3)}.
\label{eq:alpha_k_AppB}
\end{align}
The required lag-wise derivatives are
\begin{align}
\frac{\partial\alpha_k}{\partial\kappa}
&=
-\zeta_k\alpha_k,
\label{eq:dalpha_dkappa_AppB}
\\
\frac{\partial\alpha_k}{\partial\theta}
&=
\kappa\alpha_k
\frac{kd_0\sin\theta}
{\big(a_3d_k(\theta)+1\big)^2},
\label{eq:dalpha_dtheta_AppB}
\\
\frac{\partial\alpha_k}{\partial a_3}
&=
\kappa\alpha_k
\frac{d_k^2(\theta)}
{\big(a_3d_k(\theta)+1\big)^2}.
\label{eq:dalpha_da3_AppB}
\end{align}
For $x\in\{\theta,\kappa,a_3\}$, define $\mathbf{A}_x\in\mathbb{R}^{M\times M}$ by
\begin{equation}
[\mathbf{A}_x]_{m\ell}
=
\begin{cases}
0, & m=\ell,\\[1mm]
\dfrac{\partial\alpha_{|m-\ell|}}{\partial x}, & m\neq\ell.
\end{cases}
\label{eq:Ax_def_AppB}
\end{equation}
Then
\begin{equation}
\frac{\partial\mathbf{R}_b}{\partial x}
=
2\lambda_{11}\mathbf{A}_x,
\qquad
x\in\{\theta,\kappa,a_3\}.
\label{eq:dRb_dx_AppB}
\end{equation}
Using~\eqref{eq:Ry_PI_model}, the parameter sensitivities are
\begin{align}
\frac{\partial\mathbf{R}_y^{\mathrm{PI}}}{\partial\theta}
&=
\rho\left[
\frac{\partial\mathbf{S}(\theta)}{\partial\theta}
\odot\mathbf{R}_b
+
\mathbf{S}(\theta)\odot
\frac{\partial\mathbf{R}_b}{\partial\theta}
\right],
\label{eq:dRy_dtheta_PI_AppB}
\\
\frac{\partial\mathbf{R}_y^{\mathrm{PI}}}{\partial\kappa}
&=
\rho\,\mathbf{S}(\theta)\odot
\frac{\partial\mathbf{R}_b}{\partial\kappa},
\label{eq:dRy_dkappa_PI_AppB}
\\
\frac{\partial\mathbf{R}_y^{\mathrm{PI}}}{\partial a_3}
&=
\rho\,\mathbf{S}(\theta)\odot
\frac{\partial\mathbf{R}_b}{\partial a_3},
\label{eq:dRy_da3_PI_AppB}
\\
\frac{\partial\mathbf{R}_y^{\mathrm{PI}}}{\partial\rho}
&=
\mathbf{S}(\theta)\odot\mathbf{R}_b,
\label{eq:dRy_drho_PI_AppB}
\\
\frac{\partial\mathbf{R}_y^{\mathrm{PI}}}{\partial\sigma_n^2}
&=
\mathbf{I}_M.
\label{eq:dRy_dsigman_PI_AppB}
\end{align}
Here $\partial\mathbf{S}(\theta)/\partial\theta$ is given by~\eqref{eq:dS_dtheta_AppA}. Substituting the vectorized sensitivities into~\eqref{eq:FIM_SB}, followed by the Schur complement over $[\rho,\sigma_n^2]^T$, gives~\eqref{eq:CRLB_PI_matrix}.